\documentclass[reqno,a4paper, oneside, 11 pt]{amsart}

\usepackage[dvipsnames,svgnames]{xcolor}
\usepackage{graphicx} 
\usepackage{amsthm}
\usepackage{bm}
\usepackage{accents}
\usepackage{xcolor}
\usepackage{amsmath}
\usepackage[all]{xy} \SelectTips{eu}{}
\usepackage{latexsym,amsfonts,amssymb}
\usepackage{hyperref}
\hypersetup{colorlinks=true, linkcolor=blue}
\usepackage{cite}
\hypersetup{citecolor=red}
\usepackage{leftidx}
\usepackage[dvipsnames,svgnames,table]{xcolor}
\usepackage{nicefrac}
\usepackage{array}
\usepackage{mathrsfs}
\usepackage{rotating}
\usepackage{tikz-cd}
\usepackage{stmaryrd}
\usepackage{thmtools}
\usepackage{setspace}
\usepackage{geometry}
\usepackage{float}
\usepackage[utf8]{inputenc}
\usepackage[english]{babel}
\usepackage{framed}
\usepackage[dvipsnames]{xcolor}
\usepackage{tcolorbox}
\usepackage{bigints}
\usepackage{stackengine}
\DeclareMathAlphabet{\mathcal}{OMS}{cmsy}{m}{n}

\date{\today}

\newtheorem{theorem}{Theorem}[section]

\newtheorem{lemma}[theorem]{Lemma}
\newtheorem{proposition}[theorem]{Proposition}

\theoremstyle{definition}
\newtheorem{definition}[theorem]{Definition}

\newtheorem{remark}[theorem]{Remark}

\newtheorem{example}[theorem]{Example}
\AtBeginEnvironment{example}{%
  \pushQED{\qed}%
}
\AtEndEnvironment{example}{\popQED\endexample}

\newcommand{\DS}{\mathcal{D}^s}
\newcommand{\DSP}{\mathcal{D}^{s}_{+}}
\newcommand{\DSM}{\mathcal{D}^{s}_{\mu}}
\newcommand{\pl}{\bm{\varphi}}
\newcommand{\F}{\mathsf{F}}
\newcommand{\A}{\mathsf{A}}

\newcommand{\frontaffiliation}{%
  \vspace{0.2em}
  \begin{center}
    \footnotesize

    \textsuperscript{1}\,Dipartimento di Matematica ``Tullio Levi-Civita'',
    Università degli Studi di Padova,\\
    via Trieste 63, 35131 Padova, Italy\\
\textsuperscript{2}\,Gruppo Nazionale per la Fisica Matematica, Istituto Nazionale di Alta Matematica\\ ``Francesco Severi'', Sezione di Padova, Italy\\ 
      {\texttt{francesca.berlinghieri@math.unipd.it}};\quad{\texttt{giulio.giusteri@unipd.it}}
 \\
  \end{center}%
}
\makeatletter
\apptocmd{\@setauthors}{\frontaffiliation}{}{}
\makeatother

\begin{document}
\title[Inertial motion of incompressible continua]{Inertial motion of incompressible continua}
\author[Berlinghieri and Giusteri]{%
  Francesca Berlinghieri\textsuperscript{1,2},
  Giulio G. Giusteri\textsuperscript{1,2}}

\begin{abstract}
We study the inertial motion of incompressible continua within a geometric and variational framework, extending the classical Arnold--Ebin--Marsden theory from the group of volume-preserving diffeomorphisms of a fixed domain to configuration spaces of deformations with variable image.
In the latter case, the lack of a group structure requires proving some results that are instead immediate in the classical setting.
We show that the orientation-preserving deformations with suitable regularity constitute a Hilbert manifold and that volume-preserving deformations form a submanifold with tangent vectors that are mapped onto divergence-free vector fields by the Lagrangian-to-Eulerian-picture correspondence. 
In so doing, we also present the geometric structure corresponding to compressible continua.
The kinetic energy of the continuum gives rise to both a Lagrangian action, from which the equations of inertial motion are deduced, and a metric, with geodesics that are identified precisely by inertial motions.
While the inertial motion can coincide with a physical one only for incompressible perfect fluids, it can be used to provide a natural parametrization of the configuration manifold for generic continua.
We obtain a general result of local-in-time existence of solutions for the geodesic flow equation and
we present explicit examples showing that, for the same initial data, the geodesic followed by the continuum in the compressible case can leave the manifold of admissible deformations in finite time, while the corresponding incompressible geodesic exists for all times.\\[0.5em]
\end{abstract}

\maketitle

\section{Introduction and main results}
The motion of incompressible continua has long provided a fertile ground for the interplay between analysis, geometry, and mechanics. 
The geometric interpretation of the motion of incompressible perfect fluids in a fixed domain goes back to Arnold \cite{arnold1966geometrie}, who described the motions of such fluids as geodesic curves on the group of volume-preserving diffeomorphisms, that is inherently infinite-dimensional. This geometric formulation connects the analytical structure of the equations of motion with the differential geometry of the underlying configuration space and has since become a cornerstone in the modern understanding of hydrodynamics. Indeed, the study of infinite-dimensional manifolds of mappings, and in particular of
groups of diffeomorphisms, has a long and well-established history. Foundational
results on the differential structure of such spaces were developed by Leslie \cite{leslie1967differential}, and further extended by Ebin \cite{ebin1970manifold} and Omori \cite{omori1970group}, who introduced the
framework of ILH (inverse limit Hilbert) Lie groups. An alternative framework for global analysis on infinite-dimensional spaces is provided by the Convenient setting of Kriegl and Michor \cite{kriegl1997convenient}.

Ebin and Marsden \cite{ebin1970groups} further developed Arnold's approach and established that Euler equations for an incompressible perfect fluid in a fixed domain are geodesic equations on the manifold of volume-preserving diffeomorphisms of Sobolev class, that is a Hilbert manifold and also a topological group, that can be endowed with a Riemannian manifold structure induced by the right-invariant metric associated to the kinetic energy, and whose tangent vectors are represented by divergence-free velocity fields (see also \cite{ebin2015groups} for a more recent perspective on these ideas). This structure allows one to recast the analytical problem of fluid motion into a geometric problem of finding geodesics, to which the methods of global analysis can be applied.

While the Arnold--Ebin--Marsden approach was developed for the study of incompressible perfect fluids that flow within a fixed domain, it is important to understand whether analogous geometric and variational principles can be applied to a more general framework, where the configuration of the body is described by a generic volume-preserving deformation of a material domain $\Omega_0\subset\mathbb{R}^3$ into the ambient space, rather than a diffeomorphism of a fixed domain onto itself. 
In contrast with the classical framework, this configuration space does not carry a natural group structure, since compositions of admissible deformations are not well-defined. 
A closely related covariant derivation of the equations of motion for continua with variable current configurations was analyzed, in order to describe the dynamics of residually-stressed elastic bodies, by Kupferman, Olami and Segev \cite{Segev2017}. 

The main aim of the present work is to investigate, within the Sobolev setting, the geodesic flow on the configuration manifold of continua undergoing generic volume-preserving deformations, thereby extending the Arnold--Ebin--Marsden approach to a situation that can be relevant both in fluid and solid mechanics.
In such a broader setting, the Lagrangian picture provides a natural description of the motion in coordinates on the material domain that remains fixed in time, while the domain of the Eulerian formulation may constantly change. While the classical works deal with diffeomorphisms on Riemannian manifolds, we only discuss the case of Euclidean manifolds to simplify the presentation, but we expect that all results can be easily generalized to the Riemannian case.


The mathematical structure underlying this problem exhibits both geometric and analytic challenges. 
On one hand, one needs to characterize the manifold of admissible deformations, showing that it inherits a differentiable structure compatible with the Sobolev topology of the underlying space of mappings. To this end, we consider a material manifold $\Omega_0 \subset \mathbb{R}^3$, that is a bounded connected domain with Lipschitz boundary, and $s\ge 3$ integer, and define the configuration space of Sobolev orientation-preserving deformations
\begin{equation*}
\begin{aligned}
   \DSP := \{ &\pl \in H^s(\Omega_0, \mathbb{R}^3) \mid \pl \text{ is a bijection onto } \pl(\Omega_0),\\
   &\pl^{-1} \in C^1(\pl(\overline{\Omega}_0), \mathbb{R}^3) \text{ and } \det\nabla\pl>0\}.
\end{aligned}
\end{equation*}
In Section \ref{sec:manfomap}, we show that $\DSP$ is open in $H^s(\Omega_0, \mathbb{R}^3)$ and, consequently, $\DSP$ is a Hilbert manifold modeled on $H^s(\Omega_0, \mathbb{R}^3)$. 
We then prove that, despite the absence of a group structure, the stability of Sobolev regularity under inversion persists.
Namely, we show that
\begin{equation*}
\begin{aligned}
\DSP
\subseteq
\Big\{
&\pl \in H^s(\Omega_0, \mathbb{R}^3) \mid \pl \text{ is a bijection onto } \pl(\Omega_0) \text{, }\\
&\det\nabla\pl(\bm{X})>0 \text{ and } \pl^{-1}\in H^s(\pl(\Omega_0),\mathbb R^3)
\Big\},
\end{aligned}
\end{equation*}
which implies a gain in regularity of inverse mappings that is key in connecting the regularity of Eulerian fields to that of Lagrangian fields.


In Section \ref{sec:flow_equation}, we derive the Lagrangian form of the inertial motion of an incompressible continuum, that represents the geodesic flow on 
\[
\DSM
=
\{\bm{\varphi} \in \DSP \mid \det\nabla \bm{\varphi} =1\}\,.
\]
While in the case of an incompressible perfect fluid, the inertial motion coincides with the real motion in the absence of external forces or tractions, we stress that it can be used even for elastic solids to provide a (local) parametrization of the manifold of admissible deformations.

Section \ref{sec:geodesics} is devoted to proving that $\mathcal{D}^s_{\mu}$ is a submanifold of $\mathcal{D}^s_+$, with
\[
T_{\bm{\varphi}}\mathcal{D}^s_\mu=\{\bm{v}\circ\bm{\varphi}\in H^s(\Omega_0,\mathbb{R}^3)\mid \operatorname{div}\bm{v} = 0\}.
\]
The latter characterization highlights the usefulness of Eulerian fields in providing an easier description of incompressibility.
The geometric identification of inertial motions as geodesics opens the way to a local well-posedness analysis, along the lines of the works of Ebin and Marsden \cite{ebin1970groups} and Ebin \cite{ebin2015groups}. 

The global-in-time existence of geodesic flows for incompressible continua remains an open question. While in the case of compressible continua it is easy to exhibit examples in which the geodesic flow on $\DSP$ hits the manifold boundary in a finite time, the incompressibility constraint prevents that type of singularity and allows to prove, in specific cases, that the geodesic exists for all times. In Section \ref{sec:examples}, we present examples illustrating both scenarios.

In summary, the results developed in this work show that, in the case of an incompressible continuum that can deform arbitrarily in space, we can apply the geometric features of the classical Arnold--Ebin--Marsden framework, in spite of the absence of an underlying group structure. 
In our setting, the Sobolev regularity of inverse deformations allows to pass
consistently between the Lagrangian and Eulerian descriptions, characterize the
tangent spaces to $\DSM$ in terms of divergence-free Eulerian fields, and identify
inertial motions with geodesics of the kinetic-energy metric. The resulting
geodesic flow admits a locally well-posed formulation and could therefore be exploited to provide a natural parametrization of the configuration manifold for generic continua, beyond the classical setting of incompressible perfect fluids. 
We plan to investigate in future works the possible use of such a construction in the context of nonlinear elasticity.

\section{Preliminaries on manifolds of mappings}\label{sec:manfomap}

It is well known that the vector space $H^s(\Omega_0, \mathbb{R}^3)$ of Sobolev functions with domain $\Omega_0\subset \mathbb{R}^3$ has a natural structure of infinite-dimensional Hilbert manifold modeled on itself. 
For the purposes of this work, we restrict our attention to the space of Sobolev functions over $\Omega_0$ that are invertible onto their image and with regular inverse. In this section, we show that this space is an open subset of $H^s(\Omega_0, \mathbb{R}^3)$ and thus inherits the structure of differentiable manifold modeled on $H^s(\Omega_0, \mathbb{R}^3)$. Moreover, we prove that Sobolev regularity is transferred to the inverse mappings. This property will play a key role in relating Lagrangian and Eulerian fields in the subsequent analysis.

\begin{definition} Let $\Omega_0 \subset \mathbb{R}^3$ be a bounded connected domain with Lipschitz boundary and let $s$ be an integer with $s \geq 3$. We define
\[
\DS := \{ \pl \in H^s(\Omega_0, \mathbb{R}^3) \mid \pl \text{ is a bijection onto } \pl(\Omega_0) \text{ and } \pl^{-1} \in C^1(\pl(\overline{\Omega}_0), \mathbb{R}^3) \}.
\]
For any $\pl\in\DS$, we introduce the notation 
\begin{equation}
\label{def grad}
\F:=\nabla\pl
\end{equation}
and then define
\begin{equation}
\label{DSP}
\DSP:=\Big\{\pl\in \DS:\ \det\F(\bm{X})>0\ \text{for all }\bm{X}\in\overline{\Omega}_0\Big\}.
\end{equation}
\end{definition}

\begin{proposition}\label{prop:Ds1}
Let $\Omega_0 \subset \mathbb{R}^3$ be a bounded connected domain with Lipschitz boundary and let $s$ be an integer with $s \geq 3$. 
Then $\DS$ coincides with the intersection of $H^s(\Omega_0, \mathbb{R}^3)$ and the space of $C^1$ embeddings from $\overline{\Omega}_0$ into $\mathbb{R}^3$, namely
\[
\DS = H^s(\Omega_0, \mathbb{R}^3) \;\cap\; \operatorname{Emb}^1(\overline{\Omega}_0,\mathbb{R}^3).
\]
\end{proposition}

\begin{proof}
Let us denote 
\[
\mathcal{E} := H^s(\Omega_0, \mathbb{R}^3) \;\cap\; \operatorname{Emb}^1(\overline{\Omega}_0,\mathbb{R}^3).
\]
We first prove the inclusion $\DS \subseteq \mathcal{E}$. Let $\pl \in \DS$. Then we have $\pl \in H^s(\Omega_0, \mathbb{R}^3)$ and $\pl^{-1} \in C^1(\pl(\overline{\Omega}_0), \mathbb{R}^3)$. By the Sobolev embedding theorem for $s > 5/2$, we have
\[
H^s(\Omega_0, \mathbb{R}^3) \hookrightarrow C^1(\overline{\Omega}_0, \mathbb{R}^3),\]
so $\pl\in C^1(\overline{\Omega}_0, \mathbb{R}^3)$.  
Moreover, since $\pl^{-1} \circ \pl = \mathrm{Id}$, the chain rule gives
\[
(\nabla\pl^{-1})(\pl(\bm{X}))\F(\bm{X}) = \mathsf{I} \quad \forall \bm{X} \in \overline{\Omega}_0,
\]
with $\mathsf{I}$ the identity matrix, so $\F(\bm{X})$ is invertible at any point of $\overline{\Omega}_0$. Combined with the injectivity of $\pl$, this shows that $\pl$ is a $C^1$ embedding. Hence $\pl \in \mathcal{E}$.

Now, let us prove the other inclusion $\mathcal{E} \subseteq \DS$. Let $\pl \in \mathcal{E}$. Then $\pl \in H^s(\Omega_0, \mathbb{R}^3)$ and $\pl$ is a $C^1$ embedding, i.e., $\pl$ is injective, $\F(\bm{X})$ is invertible for all $\bm{X} \in \overline{\Omega}_0$, and $\pl$ is a homeomorphism onto its image. Thus $\pl^{-1}\in C^0(\pl(\overline{\Omega}_0), \mathbb{R}^3)$. By the $C^1$ inverse function theorem, $\pl^{-1}$ is locally $C^1$. Since $\pl$ is by definition globally invertible and with continuous inverse, $\pl^{-1}$ extends to a $C^1$ map on all of $\pl(\overline{\Omega}_0)$ (see \cite{lee2013smooth}). Therefore, $\pl \in \DS$.
\end{proof}

\begin{proposition}
\label{closures}
Let $\Omega_0 \subset \mathbb{R}^3$ be a bounded connected domain with Lipschitz boundary and let 
$\pl \in H^s(\Omega_0,\mathbb{R}^3)$ with $s \geq 3$ integer. If we set $\Omega=\pl(\Omega_0)$, then
\[
\overline{\Omega} = \pl({\overline{\Omega}}_0).
\]
\end{proposition}

\begin{proof}
Since $s>5/2$ and $\Omega_0$ has Lipschitz boundary, the Sobolev embedding theorem 
guarantees that $\pl$ is continuous on ${\overline{\Omega}}_0$.
The inclusion $\pl({\overline{\Omega}}_0) \subseteq \overline{\pl(\Omega_0)}$ follows directly from the continuity of $\pl$ (see \cite{munkres2000}).

Since ${\overline{\Omega}}_0$ is compact and $\pl$ is 
continuous, the image $\pl({\overline{\Omega}}_0)$ is compact, and hence closed. Moreover, 
$\pl(\Omega_0)\subseteq \pl({\overline{\Omega}}_0)$. As $\pl({\overline{\Omega}}_0)$ is a closed set containing 
$\pl(\Omega_0)$, by definition of closure it follows that
\(
\overline{\pl(\Omega_0)} \subseteq \pl({\overline{\Omega}}_0).
\)
\end{proof}

\begin{proposition}
Let $\Omega_0 \subset \mathbb{R}^3$ be a bounded connected domain with Lipschitz boundary and let $s$ be an integer with $s \geq 3$. Then $\DS$ 
is open in $H^s(\Omega_0, \mathbb{R}^3)$.
\end{proposition}

\begin{proof}
By Proposition~\ref{prop:Ds1}, we know that
\[
\DS = H^s(\Omega_0, \mathbb{R}^3) \;\cap\; \operatorname{Emb}^1(\overline{\Omega}_0,\mathbb{R}^3).
\]
It is a standard result that the set of $C^1$ embeddings of a compact domain $M$ into $\mathbb{R}^3$ is an open subset of $C^1(M, \mathbb{R}^3)$ in the $C^1$-topology (see \cite{hirsch1976differential}). Since $\overline{\Omega}_0$ is closed and bounded in $\mathbb{R}^3$, and hence compact, we thus have that $\operatorname{Emb}^1(\overline{\Omega}_0,\mathbb{R}^3)$ is open in $C^1(\overline{\Omega}_0, \mathbb{R}^3)$.

Since $s>5/2$, the Sobolev embedding theorem gives a continuous inclusion
\[
H^s(\Omega_0, \mathbb{R}^3) \hookrightarrow C^1(\overline{\Omega}_0, \mathbb{R}^3).
\]
Hence, if $\operatorname{Emb}^1(\overline{\Omega}_0,\mathbb{R}^3)$ is open in $C^1(\overline{\Omega}_0, \mathbb{R}^3)$, its preimage under the continuous inclusion is open in $H^s(\Omega_0, \mathbb{R}^3)$. But this preimage is precisely
\[
\DS = H^s(\Omega_0, \mathbb{R}^3) \cap \operatorname{Emb}^1(\overline{\Omega}_0,\mathbb{R}^3),
\]
and therefore $\DS$ is open in $H^s(\Omega_0, \mathbb{R}^3)$.
\end{proof}

\begin{theorem}
\label{thm: open}
Let $\Omega_0 \subset \mathbb{R}^3$ be a bounded connected domain with Lipschitz boundary and let $s$ be an integer with $s \geq 3$. Then $\DS$
is a Hilbert manifold modeled on $H^s(\Omega_0, \mathbb{R}^3)$. In particular, $\DS$ is a (Hilbert) submanifold of $H^s(\Omega_0, \mathbb{R}^3)$.
\end{theorem}

\begin{proof}
It is a standard fact in the theory of differentiable manifolds that any open subset of a Hilbert space naturally inherits the structure of a Hilbert manifold modeled on the same Hilbert space (see \cite{lee2013smooth}).
From the previous proposition, we know that $\DS$ is open in $H^s(\Omega_0, \mathbb{R}^3)$. Since $H^s(\Omega_0, \mathbb{R}^3)$ is a Hilbert space, the openness of $\DS \subset H^s(\Omega_0, \mathbb{R}^3)$ implies that the single chart $(\DS, Id_{\DS})$ defines a differentiable structure on $\DS$. Therefore, $\DS$ carries the structure of a Hilbert manifold modeled on $H^s(\Omega_0, \mathbb{R}^3)$.
\end{proof}

As a direct consequence of the result above, the openness property is preserved under the additional orientation–preserving constraint.

\begin{theorem}
\label{orientation-preserving open}
Let $\Omega_0\subset\mathbb R^3$ be a bounded connected domain with Lipschitz boundary and let $s$ be an integer with $s \geq 3$. Then the set $\DSP$, defined in \eqref{DSP}, is open in $H^s(\Omega_0,\mathbb R^3)$, and hence a submanifold modeled on $H^s(\Omega_0,\mathbb R^3)$.
\end{theorem}

\begin{proof}
Fix $\pl\in \DSP$. Since $s>\frac52$, thanks to the continuous Sobolev embedding
\begin{equation}
\label{sobolev embedding}
H^s(\Omega_0,\mathbb R^3)\hookrightarrow C^1(\overline{\Omega}_0,\mathbb R^3),
\end{equation}
the map $\pl\mapsto\F$ is continuous from $H^s(\Omega_0,\mathbb{R}^3)$ to
$C^0(\overline{\Omega}_0,\operatorname{Mat}_{3}(\mathbb{R}))$, where $\operatorname{Mat}_{3}(\mathbb{R})$ is the space of $3\times 3$ matrices with real coefficients.
We set
\[
\mathcal U:=\Big\{f\in C^0(\overline{\Omega}_0, \mathbb{R}) : f(\bm{X})>0\ \text{for all } \bm{X}\in\overline{\Omega}_0\Big\}
\]
and note that $\mathcal U$ is open in $C^0(\overline{\Omega}_0, \mathbb{R})$ equipped with the
$\|\cdot\|_{L^\infty}$ norm. Indeed, if $f\in\mathcal U$, then by compactness
there exists $m:=\min_{\overline{\Omega}_0} f>0$. If $g\in C^0(\overline{\Omega}_0, \mathbb{R})$
satisfies $\|g-f\|_{L^\infty}<m/2$, then, for every $\bm{X}\in\overline{\Omega}_0$,
\[
g(\bm{X})=f(\bm{X})+(g(\bm{X})-f(\bm{X})),
\]
and therefore
\[
g(\bm{X})\ge f(\bm{X})-|g(\bm{X})-f(\bm{X})|
\ge m-\frac{m}{2}
=\frac{m}{2}>0,
\]
so $g\in\mathcal U$.

Next, consider the map
\[
\Phi:H^s(\Omega_0,\mathbb R^3)\to C^0(\overline{\Omega}_0, \mathbb{R}),\qquad
\Phi(\pl):=\det\F.
\]
Since the determinant is a polynomial in the matrix entries, the map
$\A\mapsto \det \A$ is continuous from
$C^0(\overline{\Omega}_0, \operatorname{Mat}_{3}(\mathbb{R}))$ to $C^0(\overline{\Omega}_0,\mathbb{R})$.
By composition with the continuous embedding \eqref{sobolev embedding}, it follows
that $\Phi$ is continuous from $H^s(\Omega_0, \mathbb{R}^3)$ to $C^0(\overline{\Omega}_0, \mathbb{R})$.

Therefore $\Phi^{-1}(\mathcal U)$ is open in $H^s(\Omega_0,\mathbb R^3)$.
Finally, observe that
\[
\Phi^{-1}(\mathcal U)=\Big\{\\\pl\in H^s(\Omega_0,\mathbb R^3):\ \det\F(\bm{X})>0
\ \text{for all }\bm{X}\in\overline{\Omega}_0\Big\},
\]
hence $\DSP=\DS\cap \Phi^{-1}(\mathcal U)$ is open in $H^s(\Omega_0,\mathbb R^3)$
because $\DS$ is open in $H^s(\Omega_0,\mathbb R^3)$ by Theorem \ref{thm: open}.
\end{proof}

Although for every $\pl\in\DSP$, we have that $\pl^{-1}\in C^1(\overline{\Omega},\mathbb{R}^3)$ by definition, the transfer of regularity between Lagrangian and Eulerian vector fields, as we shall see in the next section, requires the inverse map to possess the same Sobolev regularity as $\pl$. In the classical fixed-domain setting, this stability under inversion is part of the standard theory of diffeomorphism groups. Here, however, $\DSP$ does not carry a group structure and the domain of $\pl^{-1}$ depends on $\pl$ itself. The following theorem shows that this regularity property is nevertheless preserved in the present setting. The proof is rather technical, but it is essentially based on a natural recursive analysis of the higher-order derivatives of the inverse map.

\begin{theorem}
\label{regularity inverse}
Let $\Omega_0\subset\mathbb R^3$ be a bounded connected domain with Lipschitz boundary and
let $s$ be an integer with $s \geq 3$. Then, for every $\pl$ in $\DSP$, we have that $\pl^{-1}$ gains regularity and belongs to $H^s(\pl(\Omega_0),\mathbb R^3)$.

\end{theorem}

\begin{proof}

Let $\pl\in \DSP$. Since $s>\frac52$, $\pl$ is a $C^1$-diffeomorphism between
$\Omega_0$ and $\Omega$.
Moreover, the positivity of the Jacobian determinant implies that
$\det\F$ is bounded away from zero on the compact set
$\overline{\Omega}_0$. Consequently,
$\F^{-1}$ is continuous and bounded on $\overline{\Omega}_0$. In fact, since $\pl\in C^1(\overline{\Omega}_0,\mathbb{R}^3)$, its gradient $\F\in C^{0}(\overline{\Omega}_0,\operatorname{Mat}_3(\mathbb{R}))$.
The function $\bm{X}\mapsto \det \F(\bm{X})$ is continuous, being a polynomial in the
entries of $\F(\bm{X})$. By assumption $\det \F(\bm{X})>0$ for all $\bm{X}\in\overline{\Omega}_0$.
Hence, by compactness, it has a strictly positive minimum
\[
m:=\min_{X\in\overline{\Omega}_0} \det \F(\bm{X}) >0,
\]
so that
\[
\det \F(\bm{X})\ge m>0 \qquad \forall\,\bm{X}\in\overline{\Omega}_0.
\]
Hence, the function
\[
\bm{X}\longmapsto \frac{1}{\det \F(\bm{X})}
\]
is continuous and bounded on $\overline{\Omega}_0$. Moreover, since $\F(\cdot)$ is continuous and each entry of the cofactor matrix $\operatorname{cof} \F$ is a polynomial in the entries of $\F$, the map
\[
\bm{X}\longmapsto \operatorname{cof} \F(\bm{X})
\]
is continuous, and hence bounded, on the compact set $\overline{\Omega}_0$. For every $\bm{X}\in\overline{\Omega}_0$ we have $\det \F(\bm{X})\neq 0$, hence $\F(\bm{X})$ is invertible
and the function
\[
\bm{X}\longmapsto \F(\bm{X})^{-1} = \frac{(\operatorname{cof} \F(\bm{X}))^{\top}}{\det \F(\bm{X})}
\]
belongs to $C^0(\overline{\Omega}_0,\operatorname{Mat}_3(\mathbb R))$, in light of the observations above. Since $\F^{-1}$ is continuous on the compact set $\overline{\Omega}_0$,
it is bounded, that is
\begin{equation*}
\|\F^{-1}\|_{L^\infty(\Omega_0)} = \sup_{\bm{X}\in\overline{\Omega}_0} |\F^{-1}(\bm{X})| <\infty.
\end{equation*}
Note also that $\pl^{-1}\in L^2(\Omega,\mathbb{R}^3)$, since $\pl^{-1}\in C^1(\overline{\Omega}, \mathbb{R}^3)$ and hence $\pl^{-1}\in L^{\infty}(\Omega, \mathbb{R}^3)$, and
$\Omega$ is a bounded domain.
We now want to prove that $\pl^{-1}\in H^s(\Omega,\mathbb{R}^3)$ by estimating its derivatives.

\smallskip
\noindent
\underline{\emph{Step 1: first derivative.}}
Since $\pl^{-1} \circ \pl = \mathrm{Id}$, the chain rule yields
\begin{equation}
\label{first derivative}
(\nabla_{\bm{x}}\pl^{-1})(\pl(\bm{X}))\F(\bm{X}) = \mathsf{I} \quad \forall \bm{X} \in \overline{\Omega}_0.
\end{equation}
Equation \eqref{first derivative} leads to the identity
\begin{equation}
\label{gradient inverse}
\nabla_{\bm{x}}\pl^{-1}(\bm{x})
=
\F^{-1}(\pl^{-1}(\bm{x})),
\end{equation}
where $\bm{x}:=\pl(\bm{X})$. Exploiting the previous identity and a change of variables, we obtain
\begin{equation*}
\|\nabla_{\bm{x}}\pl^{-1}\|_{L^2(\Omega)} = \int_{\Omega}|\nabla_{\bm{x}}\pl^{-1}(\bm{x})|^2 d\bm{x} = \int_{\Omega_0}J(\bm{X})|\F^{-1}(\bm{X})|^2d\bm{X},
\end{equation*}
where $J:=\det\F$. Since, as previously observed, both $J$ and $\F^{-1}$ are bounded in $\Omega_0$, we conclude that $\nabla_{\bm{x}}\pl^{-1}\in L^2(\Omega,\mathbb{R}^3)$.

Note that one can also trivially conclude that $\nabla_{\bm{x}}\pl^{-1}\in L^2(\Omega,\mathbb{R}^3)$ by observing that $\nabla_{\bm{x}}\pl^{-1}\in C^0(\overline{\Omega},\mathbb{R}^3)$, and hence $\nabla_{\bm{x}}\pl^{-1}$ is bounded in the bounded domain $\Omega$. However, the estimate of the derivatives of $\pl^{-1}$ arising from identity \eqref{first derivative} will be crucial for the higher-order gradients.

\smallskip
\noindent
\underline{\emph{Step 2: second derivative.}}
We start from the identity \eqref{first derivative}, i.e., in components,
\begin{equation}
\label{first derivative_comp}
\frac{\partial X_i}{\partial x_k}(\pl(\bm{X}))\,\frac{\partial x_k}{\partial X_j}(\bm{X}) = \delta_{ij},
\qquad i,j=1,2,3,
\end{equation}
where, with a slight abuse, we introduced the notation $X_i = \varphi^{-1}_i$ and $x_k = \varphi_k$. Differentiating \eqref{first derivative_comp} with respect to $X_h$, one gets
\begin{equation}
\label{second derivative_comp}
\frac{\partial}{\partial X_h}\!\left(\frac{\partial X_i}{\partial x_k}(\pl(\bm{X}))\right)\frac{\partial x_k}{\partial X_j}(\bm{X})
\;+\;
\frac{\partial X_i}{\partial x_k}(\pl(\bm{X}))\,\frac{\partial^2 x_k}{\partial X_h\partial X_j}(\bm{X})
=0.
\end{equation}
Since $\bm{x}=\pl(\bm{X})$, by the chain rule the first term of the previous equation becomes
\begin{equation}
\label{chain rule first term}
\frac{\partial}{\partial X_h}\!\left(\frac{\partial X_i}{\partial x_k}(\pl(\bm{X}))\right)
=
\frac{\partial^2 X_i}{\partial x_\ell\,\partial x_k}(\pl(\bm{X}))\,
\frac{\partial x_\ell}{\partial X_h}(\bm{X}).
\end{equation}
Plugging \eqref{chain rule first term} into \eqref{second derivative_comp} we obtain
\begin{equation}
\label{second derivative_components}
\frac{\partial^2 X_i}{\partial x_\ell\,\partial x_k}(\pl(\bm{X}))\,
\frac{\partial x_\ell}{\partial X_h}(\bm{X})\,
\frac{\partial x_k}{\partial X_j}(\bm{X})
\;+\;
\frac{\partial X_i}{\partial x_k}(\pl(\bm{X}))\,\frac{\partial^2 x_k}{\partial X_h\partial X_j}(\bm{X})
=0.
\end{equation}

Recalling that
\[
F_{kj}(\bm{X})=\frac{\partial x_k}{\partial X_j}(\bm{X}),\quad (\F^{-1})_{ik}(\bm{X})=\frac{\partial X_i}{\partial x_k}(\pl(\bm{X})),\quad (\partial_{X_h}F)_{kj}(\bm{X})=\frac{\partial^2 x_k}{\partial X_h\partial X_j}(\bm{X}),
\]
equation \eqref{second derivative_components} becomes
\begin{equation}
\label{derivative F}
\frac{\partial^2 X_i}{\partial x_\ell\,\partial x_k}(\pl(\bm{X}))\,F_{\ell h}(\bm{X})\,F_{kj}(\bm{X})
\;+\;
(\F^{-1})_{ik}(\bm{X})\,(\partial_{X_h}F)_{kj}(\bm{X})
=0.
\end{equation}
Multiplying \eqref{derivative F} on the right by $(\F^{-1})_{jm}(\bm{X})$ and using $F_{kj}(\F^{-1})_{jm}=\delta_{km}$ yields
\begin{equation}
\label{second gradient inverse}
\frac{\partial^2 X_i}{\partial x_\ell\,\partial x_m}(\pl(\bm{X}))\,F_{\ell h}(\bm{X})
=
-(\F^{-1})_{ik}(\bm{X})\,(\partial_{X_h}F)_{kj}(\bm{X})\,(\F^{-1})_{jm}(\bm{X}).
\end{equation}
Finally, multiplying \eqref{second gradient inverse} by $(\F^{-1})_{h\ell}(\bm{X})$ gives the explicit expression for $\nabla^2\pl^{-1}$ evaluated along $\pl$:
\begin{equation}
\label{Hessian inverse formula}
\left[\nabla_{\bm{x}}^2\pl^{-1}(\pl(\bm{X}))\right]_{im\ell}=
-(\F^{-1})_{ik}(\bm{X})\,(\nabla_{\bm{X}}F)_{kjh}(\bm{X})\,(\F^{-1})_{jm}(\bm{X})\,(\F^{-1})_{h\ell}(\bm{X}).
\end{equation}

We endow $n$th-order tensors with the Frobenius norm
\[
|\mathcal{T}|^2 := \sum_{i_1,\ldots,i_n=1}^3 T_{i_1\cdots i_n}^2.
\]
From \eqref{Hessian inverse formula} and the submultiplicativity of the Frobenius norm, one has
\begin{equation}
\label{submultipl}
\big|\nabla_{\bm{x}}^2\pl^{-1}(\pl(\bm{X}))\big|
\;\le\;|\F^{-1}(\bm{X})|^3\,|\nabla_{\bm{X}} \F(\bm{X})|.
\end{equation}
In particular, exploiting the previous pointwise estimate and a change of variables, one gets
\begin{equation*}
\begin{aligned}
\|\nabla_{\bm{x}}^2\pl^{-1}\|_{L^2(\Omega)}^2
&=
\int_{\Omega}\big|\nabla_{\bm{x}}^2\pl^{-1}(\bm{x})\big|^2\,d\bm{x}
=
\int_{\Omega_0}J(\bm{X})\big|\nabla_{\bm{x}}^2\pl^{-1}(\pl(\bm{X}))\big|^2\,d\bm{X}\\
&\le \int_{\Omega_0}J(\bm{X})|\F^{-1}(\bm{X})|^6\,|\nabla_{\bm{X}}\F(\bm{X})|^2\,d\bm{X}.
\end{aligned}
\end{equation*}
Since, as previously observed, both $J$ and $\F^{-1}$ are bounded in $\Omega_0$, we conclude that
\begin{equation*}
\|\nabla_{\bm{x}}^2\pl^{-1}\|^2_{L^2(\Omega)}
\;\le\;\|J\|_{L^{\infty}(\Omega_0)}\;\|\F^{-1}\|_{L^\infty(\Omega_0)}^6\,\|\nabla_{\bm{X}}\F\|^2_{L^2(\Omega_0)},
\end{equation*}
hence $\nabla_{\bm{x}}^2\pl^{-1}\in L^2(\Omega,\mathbb{R}^{27})$, as $\nabla_{\bm{X}}\F \in L^2(\Omega_0,\mathbb{R}^{27})$ because $\pl\in H^s(\Omega_0,\mathbb{R}^3)$.

\smallskip
\noindent
\underline{\emph{Step 3: third derivative.}}
Let
\begin{align*}
A_{ik}(\bm{X})&:=(\F^{-1})_{ik}(\bm{X})=\frac{\partial X_i}{\partial x_k}(\pl(\bm{X})),\\
H_{im\ell}(X)&:=\left[\nabla_{\bm{x}}^2\pl^{-1}(\pl(\bm{X}))\right]_{im\ell}=\frac{\partial^2 X_i}{\partial x_m\,\partial x_\ell}(\pl(\bm{X})).
\end{align*}
From the identity \eqref{Hessian inverse formula} one has
\begin{equation}
\label{Hess_inverse_comp}
H_{im\ell}(\bm{X})
=
-\,A_{ik}(\bm{X})\,(\nabla_{\bm{X}}F)_{kjh}(\bm{X})\,A_{jm}(\bm{X})\,A_{h\ell}(\bm{X}).
\end{equation}
By the chain rule we obtain
\begin{equation}
\label{eq:chain_third}
\frac{\partial H_{im\ell}}{\partial X_p}(\bm{X})
=
\frac{\partial^3 X_i}{\partial x_s\,\partial x_m\,\partial x_\ell}(\pl(\bm{X}))\,
\frac{\partial x_s}{\partial X_p}(\bm{X})
=
\left[\nabla_{\bm{x}}^3\pl^{-1}(\pl(\bm{X}))\right]_{im\ell s}(\bm{X})\,F_{sp}(\bm{X}).
\end{equation}
Multiplying \eqref{eq:chain_third} by $A_{pr}(\bm{X})$ and using $F_{sp}A_{pr}=\delta_{sr}$, we obtain
\begin{equation}
\label{isolated third gradient}
\left[\nabla_{\bm{x}}^3\pl^{-1}(\pl(\bm{X}))\right]_{im\ell r}=\frac{\partial H_{im\ell}}{\partial X_p}(\bm{X})\,A_{pr}(\bm{X}).
\end{equation}
Differentiating \eqref{Hess_inverse_comp} with respect to $X_p$, one gets
\begin{align}
\frac{\partial H_{im\ell}}{\partial X_p}
&=
-\frac{\partial A_{ik}}{\partial X_p}\,(\nabla_{\bm{X}}F)_{kjh}\,A_{jm}\,A_{h\ell}
-A_{ik}\,\frac{\partial}{\partial X_p}(\nabla_{\bm{X}}F)_{kjh}\,A_{jm}\,A_{h\ell}\nonumber\\
&\quad
-A_{ik}\,(\nabla_{\bm{X}}F)_{kjh}\,\frac{\partial A_{jm}}{\partial X_p}\,A_{h\ell}
-A_{ik}\,(\nabla_{\bm{X}}F)_{kjh}\,A_{jm}\,\frac{\partial A_{h\ell}}{\partial X_p}.
\label{eq:dH_expanded}
\end{align}
Note that
\[
\frac{\partial}{\partial X_p}(\nabla_{\bm{X}}F)_{kjh}=\frac{\partial^2 F_{kj}}{\partial X_p\partial X_h}=(\nabla_{\bm{X}}^2F)_{kjhp}.
\]
Moreover, exploiting \eqref{chain rule first term} and \eqref{second gradient inverse} for the derivatives of $\A:=\F^{-1}$, we have
\begin{equation}
\label{derivative A}
\frac{\partial A_{ik}}{\partial X_p}
=
- A_{ia}\,\frac{\partial F_{ab}}{\partial X_p}\,A_{bk}
=
- A_{ia}\,(\nabla_{\bm{X}}F)_{abp}\,A_{bk}.
\end{equation}
Plugging \eqref{derivative A} into \eqref{eq:dH_expanded} and \eqref{isolated third gradient} yields an explicit expression for $\nabla_{\bm{x}}^3\pl^{-1}(\pl(\bm{X}))$:
\begin{equation}
\label{third gradient inverse}
\begin{aligned}
&\left[\nabla_{\bm{x}}^3\pl^{-1}(\pl(\bm{X}))\right]_{im\ell r}\\
&=
(\F^{-1})_{ia}\,(\nabla_{\bm{X}}F)_{abp}\,(\F^{-1})_{bk}\,
(\nabla_{\bm{X}}F)_{kjh}\,(\F^{-1})_{jm}\,(\F^{-1})_{h\ell}\,(\F^{-1})_{pr}\\
&- (\F^{-1})_{ik}\,(\nabla_{\bm{X}}^2F)_{kjhp}\,(\F^{-1})_{jm}\,(\F^{-1})_{h\ell}\,(\F^{-1})_{pr}\\
&+ (\F^{-1})_{ik}\,(\nabla_{\bm{X}}F)_{kjh}\,
(\F^{-1})_{ja}\,(\nabla_{\bm{X}}F)_{abp}\,(\F^{-1})_{bm}\,(\F^{-1})_{h\ell}\,(\F^{-1})_{pr}\\
&+ (\F^{-1})_{ik}\,(\nabla_{\bm{X}}F)_{kjh}\,(\F^{-1})_{jm}\,
(\F^{-1})_{ha}\,(\nabla_{\bm{X}}F)_{abp}\,(\F^{-1})_{b\ell}\,(\F^{-1})_{pr}.
\end{aligned}
\end{equation}
In order to find an explicit pointwise bound for $|\nabla_{\bm{x}}^3\pl^{-1}(\pl(\bm{X}))|$, we estimate each term in \eqref{third gradient inverse} exploiting the submultiplicative property of the Frobenius norm. One then gets the following pointwise bound
\begin{equation}
\label{pointwise_third}
\big|\nabla_{\bm{x}}^3\pl^{-1}(\pl(\bm{X}))\big|
\le
|\F^{-1}(\bm{X})|^4\,|\nabla_{\bm{X}}^2\F(\bm{X})|
+
3\,|\F^{-1}(\bm{X})|^5\,|\nabla_{\bm{X}}\F(\bm{X})|^2.
\end{equation}
Then, using \eqref{pointwise_third} and the fact that, for any $a$, $b$ real numbers, we have the inequality $(a+b)^2\leq 2a^2+2b^2$, we obtain
\begin{align*}
& \|\nabla_{\bm{x}}^3\pl^{-1}\|_{L^2(\Omega)}^2 =
\int_{\Omega_0}J(\bm{X})\big|\nabla_{\bm{x}}^3\pl^{-1}(\pl(\bm{X}))\big|^2\,d\bm{X} \nonumber\\
&\le 2\int_{\Omega_0}J|\F^{-1}|^8\,|\nabla_{\bm{X}}^2\F|^2\,d\bm{X}
+
18\int_{\Omega_0}J|\F^{-1}|^{10}\,|\nabla_{\bm{X}}\F|^4\,d\bm{X}.
\end{align*}
In particular, since $\F^{-1}\in L^\infty(\Omega_0,\mathbb{R}^9)$, we have
\begin{equation}
\begin{aligned}
\label{estimate third derivative}
\|\nabla_{\bm{x}}^3\pl^{-1}\|_{L^2(\Omega)}^2 &\leq  \;\|J\|_{L^{\infty}(\Omega_0)}\;\|\F^{-1}\|_{L^\infty(\Omega_0)}^8\,\|\nabla_{\bm{X}}^2\F\|_{L^2(\Omega_0)}^2\\
&+\;\|J\|_{L^{\infty}(\Omega_0)}\;\|\F^{-1}\|_{L^\infty(\Omega_0)}^{10}\,\|\nabla_{\bm{X}}\F\|_{L^4(\Omega_0)}^4.
\end{aligned}
\end{equation}
Since the fourth-order tensor $\nabla_{\bm{X}}^2\F\in L^2(\Omega_0,\mathbb{R}^{81})$, as $\pl\in H^s(\Omega_0,\mathbb{R}^3)$, the first term in \eqref{estimate third derivative} is bounded. For the second term, note that
\[
\nabla_{\bm{X}}\F\in H^{s-2}(\Omega_0,\mathbb{R}^{27})
\]
and $s-2\geq 1$, as $s\geq 3$. Hence
\[
\nabla_{\bm{X}}\F\in H^{1}(\Omega_0,\mathbb{R}^{27})
\]
Moreover, by the Sobolev embedding theorem in dimension three,
\[
H^1(\Omega_0, \mathbb{R}^{27}) \hookrightarrow L^6(\Omega_0, \mathbb{R}^{27}),
\]
and since $\Omega_0$ is bounded, the inclusion $L^6(\Omega_0, \mathbb{R}^{27}) \subseteq L^4(\Omega_0, \mathbb{R}^{27})$
also holds. Consequently, there exists a constant $C>0$ such that
\[
\|\nabla_{\bm{X}}\F\|_{L^4(\Omega_0)}
\le
C \|\nabla_{\bm{X}}\F\|_{H^1(\Omega_0)},
\]
so also the second term in \eqref{estimate third derivative} is bounded. Thus, $\nabla^3\pl^{-1}\in L^2(\Omega,\mathbb{R}^{81})$.
\smallskip

\noindent
\underline{\emph{Step 4: higher-order derivatives.}} For $s>3$, define the $(s+1)$-th order tensor
\[
\mathbb{T}_s(\bm{X}) := \nabla_x^{\,s} \pl^{-1}(\pl(\bm{X})).
\]
Note that one has the following estimate on the Frobenius norm of $\mathbb{T}_s(\bm{X})$:
\begin{equation}
\label{estimate gradient s}
|\mathbb{T}_s(\bm{X})|
\le
C_s
\sum_{\substack{k_1,\dots,k_{s-1} \ge 0\\
\sum_{j=1}^{s-1} j k_j = s-1}}
|\A(\bm{X})|^{\,s+K}
\prod_{j=1}^{s-1}
|\nabla_{\bm{X}}^{\,j} \F(\bm{X})|^{k_j},
\end{equation}
where $K := \sum_{j=1}^{s-1} k_j$ and $C_s>0$ is a constant. For $s=1,2,3$, the inequality \eqref{estimate gradient s} corresponds exactly to \eqref{gradient inverse}, \eqref{submultipl} and \eqref{pointwise_third}. Arguing by induction, we can assume now that \eqref{estimate gradient s} holds true for $s\geq 4$, and we can show that it holds also for $s+1$. In fact, recall that the chain rule yields
\begin{equation}
\label{chain rule s+1 derivative}
\mathbb{T}_{s+1}(\bm{X})=\nabla_{\bm{x}}^{s+1}\pl^{-1}(\pl(\bm{X}))=\left[\nabla_{\bm{X}}\mathbb{T}_s(\bm{X})\right]\A(\bm{X}).
\end{equation}
Recalling the inductive assumption \eqref{estimate gradient s}, when we differentiate $\mathbb{T}_s(\bm{X})$ we can encounter two cases:
\smallskip

\textbf{Case (i):} $\frac{\partial}{\partial\bm{X}}$ acts on a factor $\nabla^j_{\bm{X}}\F$.

Then
\[
\frac{\partial}{\partial\bm{X}} (\nabla^j_{\bm{X}}\F)= \nabla^{j+1}_{\bm{X}}\F.
\]
Hence in the corresponding term in \eqref{estimate gradient s} the exponent $k_j$ decreases by one and $k_{j+1}$ increases by one.
Therefore the weight increases by one:
\[
\sum_{j=1}^{s} j\,k_j
=
\left(\sum_{j=1}^{s-1} j\,k_j\right) + 1
=
(s-1)+1
=
s.
\]
In particular, derivatives up to order $s$ of $\F$ may now appear.

\smallskip

\textbf{Case (ii):} $\frac{\partial}{\partial\bm{X}}$ acts on a factor $\A$.

Using \eqref{derivative A}, each such differentiation produces an additional factor $\nabla_{\bm{X}}\F$
and an additional factor $\A$.
Hence $k_1$ increases by one and therefore the total weight is again increased by one:
\[
\sum_{j=1}^{s} j\,k_j = (s-1)+1 = s.
\]

With this in mind, recalling \eqref{chain rule s+1 derivative} and the inductive assumption \eqref{estimate gradient s}, we obtain the bound
\[
|\mathbb{T}_{s+1}(\bm{X})|
\le
C_{s+1}
\sum_{\substack{k_1,\dots,k_s \ge 0\\
\sum_{j=1}^{s} j\,k_j = s}}
|\A(\bm{X})|^{\,s+1+K'}
\prod_{j=1}^{s}
|\nabla_{\bm{X}}^{j}\F(\bm{X})|^{k_j},
\]
where $K' := \sum_{j=1}^{s} k_j$ and $C_{s+1}>0$ is a constant.

Now that \eqref{estimate gradient s} is proven, we can exploit it to estimate the $L^2$-norm of $\nabla^{s}_{\bm{x}}\pl^{-1}$. Using the change of variables $\bm{x}=\pl(\bm{X})$, and the fact that the Jacobian determinant $J$ is bounded, we have
\begin{equation*}
\|\nabla_{\bm{x}}^{s}\pl^{-1}\|_{L^2(\Omega)}^2
\le
\|J\|_{L^\infty(\Omega_0)}
\|\nabla_x^{s}\pl^{-1}(\pl(\bm{X}))\|_{L^2(\Omega_0)}^2 .
\end{equation*}
From \eqref{estimate gradient s}, since $\A$ is bounded, one has
\begin{equation*}
\|\nabla_{\bm{x}}^{s}\pl^{-1}\|_{L^2(\Omega)}^2
\le
C^2_s\|J\|_{L^\infty}\|\A\|_{L^\infty}^{2(s+K)}
\sum_{\substack{k_1,\dots,k_{s-1} \ge 0\\
\sum_{j=1}^{s-1} j k_j = s-1}}\Big\|
\prod_{j=1}^{s-1}
|\nabla_{\bm{X}}^{\,j}\F|^{k_j}
\Big\|^2_{L^2(\Omega_0)}.
\end{equation*}
Thus, it suffices to estimate
\[
\Big\|
\prod_{j=1}^{s-1}
|\nabla_{\bm{X}}^{\,j}\F|^{k_j}
\Big\|^2_{L^2(\Omega_0)} .
\]
Because $\pl \in H^{s}(\Omega_0,\mathbb{R}^3)$, we have $\nabla_{\bm{X}}^{\,j}\F \in H^{s-1-j}(\Omega_0, \mathbb{R}^{3^{j+2}})$. The Sobolev embedding theorem implies:
\begin{align*}
j \le s-3 &\quad\Rightarrow\quad
\nabla_{\bm{X}}^{\,j}\F \in L^\infty(\Omega_0, \mathbb{R}^{3^{j+2}}),
\\[4pt]
j = s-2 &\quad\Rightarrow\quad
\nabla_{\bm{X}}^{\,s-2}\F \in H^1(\Omega_0, \mathbb{R}^{3^{s}})
\hookrightarrow L^6(\Omega_0, \mathbb{R}^{3^{s}}),
\\[4pt]
j = s-1 &\quad\Rightarrow\quad
\nabla_{\bm{X}}^{\,s-1}\F \in L^2(\Omega_0, \mathbb{R}^{3^{s+1}}).
\end{align*}
The constraint
\[
\sum_{j=1}^{s-1} j k_j = s-1
\]
implies:

\begin{enumerate}
\item $k_{s-1} \in \{0,1\}$, because if $k_{s-1}\geq 2$ then $\sum j k_j \geq 2(s-1)$.

\item If $k_{s-1} =1$, then $k_j=0$ for all $j<s-1$.

\item If $k_{s-1}=0$, then $k_{s-2}\in\{0,1\}$, because if $k_{s-2}\geq 2$ then
\[
\sum j k_j \geq 2(s-2)>s-1,
\]
since $s> 3$.

\item If $k_{s-1}=0$ and $k_{s-2}=1$, then necessarily
$k_1=1$ and all other $k_j=0$.
\end{enumerate}
Hence only three cases occur.

\paragraph{\textbf{Case I}}
If $k_{s-1}=1$,
\[
\Big\|
\prod_{j=1}^{s-1}
|\nabla_{\bm{X}}^{\,j}\F|^{k_j}
\Big\|^2_{L^2(\Omega_0)} = \|\nabla_{\bm{X}}^{s-1}\F\|^2_{L^2(\Omega_0)}
\]
and therefore it is bounded, since $\nabla_{\bm{X}}^{\,s-1}\F \in L^2(\Omega_0, \mathbb{R}^{3^{s+1}})$.
\paragraph{\textbf{Case II}}
If $k_{s-1}=0$ and $k_{s-2}=1$, and hence $k_1 = 1$,
\[
\Big\|
\prod_{j=1}^{s-1}
|\nabla_{\bm{X}}^{\,j}\F|^{k_j}
\Big\|^2_{L^2(\Omega_0)} = \|(\nabla_{\bm{X}}^{s-2}\F)(\nabla_{\bm{X}}\F)\|^2_{L^2(\Omega_0)}
\]
Since $s>3$,
$\nabla_{\bm{X}} \F$ is bounded. Therefore
\[
\Big\|
\prod_{j=1}^{s-1}
|\nabla_{\bm{X}}^{\,j}\F|^{k_j}
\Big\|^2_{L^2(\Omega_0)} \leq \|\nabla_{\bm{X}}\F\|^2_{L^\infty} \|\nabla_{\bm{X}}^{s-2}\F\|^2_{L^2(\Omega_0)}
\]
and it is bounded, since $\nabla_{\bm{X}}^{s-2}\F\in H^1$.
\paragraph{\textbf{Case III}}
If $k_{s-1}= k_{s-2} = 0$, then all derivatives satisfy
$j\leq s-3$ and are bounded.
Hence
\[
\Big\|
\prod_{j=1}^{s-1}
|\nabla_{\bm{X}}^{\,j}\F|^{k_j}
\Big\|^2_{L^2(\Omega_0)} \leq |\Omega_0|\prod_{j=1}^{s-1}
\|\nabla_{\bm{X}}^{\,j}\F\|^{2k_j}_{L^\infty},
\]
where $|\Omega_0|$ is the three-dimensional Lebesgue measure of $\Omega_0$. This proves the boundedness of the left-hand side since $\Omega_0$ is a bounded domain.

Combining the previous estimates, we immediately get that, in any of the possible cases described above, we have
\begin{equation*}
\|\nabla_x^{\,s}\pl^{-1}\|_{L^2(\Omega)}< +\infty.
\end{equation*}
Therefore, $\pl^{-1}\in H^s(\Omega,\mathbb{R}^3)$ and this completes the proof.
\end{proof}

\begin{remark}
Note that, to have a more precise estimate of the Frobenius norm of $\mathbb{T}_s(\bm{X})=\nabla_{\bm{x}}^{\,s}\pl^{-1}(\pl(\bm{X}))$, we can write
\begin{equation*}
|\mathbb{T}_s(\bm{X})|
\le
\sum_{\substack{k_1,\dots,k_{s-1}\ge0\\
\sum_{j=1}^{s-1} j k_j = s-1}}
c_s(k)\,
|\A(\bm{X})|^{\,s+K}
\prod_{j=1}^{s-1}
|\nabla_{\bm{X}}^{\,j}\F(\bm{X})|^{k_j},
\qquad
K := \sum_{j=1}^{s-1} k_j,
\end{equation*}
with $c_s(k)$ non–negative constants. However, since the set of admissible multi–indices 
\[
\mathcal K_s := \{k=(k_1,\dots,k_{s-1}) : k_j\ge 0,\ \sum_{j=1}^{s-1} j k_j = s-1\}
\]
is finite, we may define
\[
C_s := \max_{k\in \mathcal K_s} c_s(k),
\]
and therefore obtain \eqref{estimate gradient s}, with all combinatorial coefficients $c_s(k)$ absorbed into a constant $C_s>0$.
\end{remark}

\section{The inertial flow equation}\label{sec:flow_equation}

In this section, we derive the equation that describes the inertial motion for an incompressible continuum in the absence of internal and external forces beyond the reactions to incompressibility.
We obtain the equation by assuming that the action associated with the kinetic energy is stationary under volume-preserving variations of the configuration.
This variational characterization yields the Lagrangian formulation of the incompressible Euler equation, with the important feature of including motions in which the spatial domain occupied by the continuum can vary.

We describe the deformation of a continuum by means of the \emph{placement} map
\[
\pl \colon 
\left\{
\begin{aligned}
\Omega_0\times\mathbb{R}&\to \mathbb{R}^3 \\
(\bm{X},t) &\mapsto \pl(\bm{X},t),
\end{aligned}\right.
\]
where $\Omega_0\subseteq \mathbb{R}^3$ is the \emph{material manifold}, described in terms of the \emph{material coordinates} $\bm{X}$, while $\Omega_t:=\pl(\Omega_0,t)$ denotes the \emph{current configuration} at time $t$, described in terms of the \emph{spatial coordinates} $\bm{x}$. The placement associates to each material point $\bm{X}$ and to each time $t\in\mathbb{R}$ the position $\pl(\bm{X},t)$ of $\bm{X}$ at time $t$.
We assume that for every $t\in \mathbb{R}$, $\pl(\cdot,t)\in \DSP$, that is, $\pl(\cdot,t) \in H^{s}(\Omega_0,\mathbb{R}^3)$, invertible, with inverse $\pl^{-1}(\cdot,t)\in H^s({\Omega}_t,\mathbb{R}^3)$. Thus $\bm{X}=\pl^{-1}(\bm{x},t)$. We recall that we are requiring $s\geq 3$ integer. In this context, we refer to the tensor field $\F$ defined in \eqref{def grad} as the \emph{deformation gradient}.
The incompressibility constraint translates into the condition $\operatorname{det}\F(\bm{X},t) = 1$ for every $\bm{X} \in \Omega_0$ and for every $t\in\mathbb{R}$. In spatial coordinates on $\Omega_t$, using Euler's identity for the derivative of $J:=\operatorname{det}\F$, one gets that this condition becomes $\operatorname{div}_{\bm{x}}\bm{v} =0$, where $\bm{v}(\bm{x},t)= \bm{V}(\pl^{-1}(\bm{x},t),t)$ is the \emph{Eulerian velocity field} and $\bm{V}(\bm{X},t)=\dot{\pl}(\bm{X},t)$ is the \emph{Lagrangian velocity field}.

The \emph{kinetic energy} of a continuum moving by $\pl(\cdot,t)$ is given by
\begin{equation}
\label{kinetic energy}
\mathcal{K}(t) = \int_{\Omega_0}\frac{1}{2}\rho_0\rvert\dot{\pl}(\bm{X},t)\lvert^2d\bm{X},
\end{equation}
where $\rho_0$ is the \emph{mass density} on the material manifold and $|\cdot|$ denotes the Euclidean norm in $\mathbb{R}^3$. We restrict our attention to the case of incompressible homogeneous materials, thus we have that $\rho_0$ does not depend on the material coordinate $\bm{X}$ and the mass density in the current configuration is $\rho(\bm{x},t) = \rho_0$ for all $\bm{x}\in\Omega_t$ and all $t\in \mathbb{R}$, i.e., the mass density is uniform in space and constant in time.

Since the placement (or deformation) at each time $t$ is represented by a point $\pl(\cdot,t)$ on the manifold $\DSP$, it is clear that the motion given by the placement map describes a curve in $\DSP$. Therefore, the corresponding velocity field naturally arises as a tangent vector to this manifold. The following proposition formalizes this observation.

\begin{proposition}
\label{tangent vectors}
Let \(\Omega_0 \subset \mathbb{R}^3\) be a bounded connected domain with Lipschitz boundary and let $s$ be an integer with $s \geq 3$. Then elements of the tangent space of $\DSP$ at a point $\pl(\cdot,t)$ consist of the Lagrangian velocity fields $\bm{V}(\cdot,t)=\bm{v}(\pl(\cdot,t),t)$.
\end{proposition}
\begin{proof}
The result follows naturally from the evolution equation of $\pl$, given by
\begin{equation*}
\frac{\partial\pl(\bm{X},t)}{\partial t} = \bm{v}(\pl(\bm{X},t),t) = \bm{V}(\bm{X},t),
\end{equation*}
with $\bm{V}(\cdot,t)\in H^s(\Omega_0,\mathbb{R}^3)$. In fact,  recalling Theorem \ref{orientation-preserving open}, we have that $H^s(\Omega_0,\mathbb{R}^3)$ can be identified with the tangent space at any point to $\DSP$ (see also \cite{lee2013smooth}).
\end{proof}

After identifying the tangent space of $\DSP$ in terms of Lagrangian velocity fields, it is useful to examine how their Eulerian counterparts inherit regularity properties. The following proposition makes this relation precise.
  
\begin{proposition}
\label{prop:regularity velocity}
Let $\Omega_0 \subset \mathbb{R}^3$ be a bounded connected domain with Lipschitz boundary and let $s$ be an integer with $s \geq 3$. 
For every $t \in \mathbb{R}$, let $\pl(\cdot, t) \in \DSP$ and denote by 
$\bm{V}(\cdot, t) = \dot{\pl}(\cdot, t)$ the Lagrangian velocity field and by $\bm{v}(\cdot, t) = \bm{V}(\pl^{-1}(\cdot, t), t)$ the corresponding Eulerian velocity field. Then $\bm{V}(\cdot, t) \in H^s(\Omega_0, \mathbb{R}^3)$ and $\bm{v}(\cdot, t) \in H^s(\Omega_t, \mathbb{R}^3)$.
\end{proposition}

The proof, that involves a recursive analysis of the derivatives of
$\bm{V}\circ\pl^{-1}$, is rather technical and we report it for completeness at the end of this section.

\begin{remark}
\label{composition diffeo}
The transfer of regularity stated in Proposition~\ref{prop:regularity velocity} is not a feature pertaining only to Eulerian velocity fields, but rather reflects a general stability property of Sobolev spaces under composition with sufficiently regular diffeomorphisms.
Indeed, with \(\Omega_0 \subset \mathbb{R}^3\) a bounded connected domain with Lipschitz boundary and $s$ an integer with $s \geq 3$, let $\pl \in \DSP$, then for any $h \in H^k(\Omega)$, with $k\in\mathbb{N}$, the composition $h \circ \pl$ belongs to $H^k(\Omega_0)$,
where $\Omega = \pl(\Omega_0)$, and the map
\[
H^k(\Omega) \times \DSP \to H^k(\Omega_0), \qquad (h,\pl) \mapsto h \circ \pl
\]
is continuous. On the other hand, for any $g \in H^k(\Omega_0)$, the composition $g \circ \pl^{-1}$ belongs to $H^k(\Omega)$, and the map
\[
H^k(\Omega_0) \times \DSP \to H^k(\Omega), \qquad (g,\pl) \mapsto g \circ \pl^{-1}
\]
is continuous.

This result was first mentioned in the $C^k$ setting (see \cite{AbrahamSmale1967} and \cite{Eells1958}) and then in the Sobolev setting in the case of $\DSP$ group of diffeomorphisms of a fixed domain $\Omega$ onto itself (see \cite{ebin1970groups}).
\end{remark}

We now show a preliminary result that will be crucial for the construction of the main equation. The following proposition is a Lagrangian formulation of the Helmholtz--Hodge decomposition (see for example \cite{chorin1993mathematical}) typically given in the Eulerian setting.
\begin{proposition}
\label{ortho}
Let \(\Omega_0 \subset \mathbb{R}^3\) be a bounded connected domain with Lipschitz boundary, let $s$ be an integer with $s \geq 3$ and let $\pl(\cdot,t)\in \DSP$, for every $t\in \mathbb{R}$. Then, under the incompressibility condition, the subspace of vector fields $\bm{W}(\cdot,t)\in H^{s}(\Omega_0,\mathbb{R}^3)$, such that the corresponding $\bm{w}(\cdot,t)=\bm{W}(\pl^{-1}(\cdot,t),t)\in H^{s}(\Omega_t,\mathbb{R}^3)$ is a divergence-free vector field, is orthogonal to the subspace of vector fields of the form $\F^{-\top}\nabla_{\bm{X}}\hat{p}(\bm{X},t)$, where $\hat{p}(\cdot,t)$ is a material scalar field with
\(
\hat{p}(\bm{X},t)=0 
\) on $\partial\Omega_0$.
\end{proposition}
\begin{proof}
It is well known that, for every $t\in \mathbb{R}$, the gradient of a scalar field $p(\cdot,t)$, with $p = 0$ on the boundary of $\Omega_t$, is orthogonal (in the $L^2$-sense) to any divergence-free vector field $\bm{w}(\cdot,t)$. In fact, let
\begin{equation}
\label{grad pressure spatial}
\bm{h}(\bm{x},t) = \nabla_{\bm{x}}p(\bm{x},t),
\end{equation}
with
\begin{equation}
\label{traction boundary spatial}
p(\bm{x},t)=0\quad\textnormal{on}~\partial\Omega_t.
\end{equation}
From the divergence theorem and the properties of the fields $\bm{w}(\cdot,t)$ and $p(\cdot,t)$, it follows that
\begin{equation}
\label{divergence theorem}
\begin{split}
&\int_{\Omega_t} \bm{w}(\bm{x},t)\cdot \bm{h}(\bm{x},t)d\bm{x}=\int_{\Omega_t} \bm{w}(\bm{x},t)\cdot\nabla_{\bm{x}}p(\bm{x},t)d\bm{x}\\
&=-\int_{\Omega_t} p(\bm{x},t)\operatorname{div}_{\bm{x}}\bm{w}(\bm{x},t)d\bm{x} + \int_{\partial\Omega_t}p(\bm{x},t)\bm{n}\cdot \bm{w}(\bm{x},t)~dS = 0,
\end{split}
\end{equation}
where $\bm{n}$ is the outward unit normal to $\partial\Omega_t$ and $dS$ denotes the integration with respect to the two-dimensional Hausdorff measure on $\partial\Omega_t$.
We now search for a vector field $\bm{H}(\cdot,t)$ on $\Omega_0$ that is orthogonal to $\bm{W}(\cdot,t)$ in the sense of $L^2$. Naming $\hat{p}(\bm{X},t)=p(\pl(\bm{X},t),t)$, from a simple change of variables in \eqref{grad pressure spatial} we obtain
\begin{equation*}
\nabla_{\bm{x}}p(\bm{x},t) = \F^{-\top}\nabla_{\bm{X}}\hat{p}(\bm{X},t).
\end{equation*}
From the boundary condition \eqref{traction boundary spatial} we get
\begin{equation*}
\hat{p}(\bm{X},t)=0\quad \textnormal{on}~\partial\Omega_0.
\end{equation*}
Thus, let $\bm{H}(\bm{X},t)=\F^{-\top}\nabla_{\bm{X}}\hat{p}(\bm{X},t)$. From equation \eqref{divergence theorem} and the incompressibility condition, it follows that
\begin{equation*}
\begin{split}
&\int_{\Omega_0}\bm{W}(\bm{X},t)\cdot\bm{H}(\bm{X},t)d\bm{X} = \int_{\Omega_0}\bm{W}(\bm{X},t)\cdot \F^{-\top}\nabla_{\bm{X}}\hat{p}(\bm{X},t)d\bm{X}\\
&=\int_{\Omega_t}\bm{W}(\pl^{-1}(\bm{x},t),t)\cdot\nabla_{\bm{x}}p(\bm{x},t)d\bm{x} = \int_{\Omega_t}\bm{w}(\bm{x},t)\cdot\nabla_{\bm{x}}p(\bm{x},t)d\bm{x} =0.
\end{split}
\end{equation*}
This completes the proof.
\end{proof}

We can now prove the main result of this section. 

\begin{theorem}
\label{th: euler lagrangian}
The equation of motion for an incompressible homogeneous continuum that moves by inertia is
\begin{equation}
\label{Lagrangian Euler equation}
\rho_0\ddot{\pl}(\bm{X},t) = -\F^{-\top}\nabla_{\bm{X}}\hat{p}(\bm{X},t),
\end{equation}
where $\hat{p}$ is the material pressure field, with boundary condition
\begin{equation}
\label{boundary cond}
\hat{p}(\bm{X},t)=0 \quad\textit{on}~\partial\Omega_0.
\end{equation}
\end{theorem}
\begin{proof}
Using the principle of least action on the kinetic energy functional, we define the motion as a stationary curve of
\[
\mathcal{F}[\pl] := \int_0^T\mathcal{K}(t)dt = \int_0^T\int_{\Omega_0}\frac{1}{2}\rho_0\rvert\dot{\pl}(\bm{X},t)\lvert^2d\bm{X}dt.
\]
Thus we have
\begin{equation}
\label{variation kinetic energy}
\delta\mathcal{F} = \frac{d}{ds}\mathcal{F}[\pl +s \bm{W}]\rvert_{s=0} = 0,
\end{equation}
for every $\bm{W}(\cdot,t)\in H^{s}(\Omega_0,\mathbb{R}^3)$ such that $\bm{W}(\cdot,0)= \bm{W}(\cdot,T) = \bm{0}$ and $\bm{w}(\cdot,t)=\bm{W}(\pl^{-1}(\cdot,t),t)\in H^{s}(\Omega_t,\mathbb{R}^3)$ is a divergence-free vector field, due to the incompressibility constraint. 
In the next section, we will show that these variations correspond precisely to tangent vectors to the submanifold of volume-preserving deformations. 
Using an integration by parts and the vanishing initial and final conditions for $\bm{W}$, equation \eqref{variation kinetic energy} can be written as
\begin{equation*}
\begin{split}
0&=\frac{d}{ds} \int_0^T\int_{\Omega_0}\frac{1}{2}\rho_0\left|\frac{\partial}{\partial t}\left(\pl(\bm{X},t)+s\bm{W}(\bm{X},t)\right)\right|^2d\bm{X}dt\rvert_{s=0}\\
&= \int_0^T\int_{\Omega_0}\rho_0\left(\dot{\pl}(\bm{X},t)\cdot \frac{d}{dt}\bm{W}(\bm{X},t)\right)d\bm{X}dt\\
&= \int_0^T\int_{\Omega_0} -\rho_0\ddot{\pl}(\bm{X},t)\cdot \bm{W}(\bm{X},t)d\bm{X}dt,
\end{split}
\end{equation*}
for every $\bm{W}$ such that $\operatorname{div}_{\bm{x}}\bm{w}=0$. With a simple localization argument, we get
\begin{equation}
\label{localization ortho}
\rho_0\ddot{\pl}(\bm{X},t) = -\bm{H}(\bm{X},t),
\end{equation}
where $\bm{H}(\cdot,t)$ is orthogonal (in the $L^2$-sense) to the subspace of vector fields $\bm{W}(\cdot,t)$ such that $\operatorname{div}_{\bm{x}}\bm{w}=0$. Recalling the previous proposition, equation \eqref{localization ortho} becomes
\begin{equation*}
\rho_0\ddot{\pl}(\bm{X},t) = -\F^{-\top}\nabla_{\bm{X}}\hat{p}(\bm{X},t) \quad\textnormal{in}~\Omega_0,
\end{equation*}
with boundary condition
\begin{equation*}
\hat{p}(\bm{X},t)=0 \quad\text{on}~\partial\Omega_0.
\end{equation*}
Equations \eqref{Lagrangian Euler equation} and \eqref{boundary cond} correspond to a Lagrangian formulation of the Euler equation for an incompressible homogeneous continuum.
\end{proof}

\begin{remark}
Since only the material pressure gradient appears in the equation of motion, the Lagrange multiplier is defined up to the addition of an arbitrary function of time. Therefore, the boundary condition \eqref{boundary cond} may equivalently be replaced by
\begin{equation*}
\hat{p}(\bm{X},t)=C(t) \quad\text{on}~\partial\Omega_0,
\end{equation*}
where $C(t)$ is an arbitrary function of time, without affecting the resulting evolution equation. Other choices of admissible variations lead to different natural boundary conditions. For instance, if one prescribes
\begin{equation*}
\bm{W}\cdot\F^{-\top}\bm{N}=0\quad \text{on}~\partial\Omega_0,
\end{equation*}
where $\bm{N}$ is the outer unit normal to $\partial\Omega_0$,
the boundary integral vanishes identically and no boundary condition is imposed on the Lagrange multiplier $\hat{p}$. This condition is equivalent to requiring the corresponding Eulerian variation to be tangent to the deformed boundary, namely
\begin{equation*}
\bm{w}\cdot\bm{n}=0\quad\text{on}~\partial\Omega_t. 
\end{equation*}
Consequently, this choice is compatible with the classical impermeability condition for incompressible fluids flowing inside a fixed domain. More generally, one can partition the boundary $\partial\Omega_0$ as the disjoint union of $\Gamma_1$ and $\Gamma_2$, imposing $\hat{p}=C(t)$ on $\Gamma_1$ and $\bm{W}\cdot\F^{-\top}\bm{N}=0$ on $\Gamma_2$, thereby allowing mixed boundary conditions.
\end{remark}
\begin{remark}
\label{rmk:inertial motion}
The variational argument leading to Theorem \ref{th: euler lagrangian} is not intrinsically restricted to the incompressible case. Indeed, if one considers arbitrary variations of the placement map, the stationarity of the Lagrangian action $\mathcal{F}$ immediately yields
\begin{equation}\label{eq:inertial_compressible}
\rho_0\ddot{\pl}=0,
\end{equation}
that is, each material point moves with constant velocity along a straight line.
This equation should not be interpreted as the physical equation of motion of a compressible continuum. In compressible materials, internal stresses arise from the interactions among the material particles and are therefore present even in the absence of the volume-preserving constraint. Consequently, the actual balance of linear momentum contains stress contributions and does not reduce to the simple inertial equation \eqref{eq:inertial_compressible}.
The incompressible case is fundamentally different. For incompressible perfect fluids, the only internal stress is the pressure enforcing the volume constraint. Therefore, when the Lagrangian action $\mathcal{F}$ is required to be stationary under incompressible variations, the inertial motion selected by the variational principle coincides with the physical motion of an incompressible perfect fluid.
\end{remark}

\begin{proof}[Proof of Proposition \ref{prop:regularity velocity}]
We have clearly $\bm{V}(\cdot,t)\in H^s(\Omega_0,\mathbb{R}^3)$ for every $t\in\mathbb{R}$, by definition of the Lagrangian velocity field and because $\DSP$ is an open subset of $H^s(\Omega_0,\mathbb{R}^3)$.
We now want to prove that $\bm{v}(\cdot,t) \in H^s(\Omega_t,\mathbb{R}^3)$ by estimating its derivatives. For simplicity of notation, we omit the explicit dependence on time.
\smallskip

\noindent
\underline{\emph{Step 0: estimate of the Eulerian velocity field.}}
From the definition of $\bm{V}$ and $\bm{v}$, and the boundedness of the Jacobian determinant $J$, we have that the change of variable $\bm{x}=\pl(\bm{X},t)$ yields
\begin{equation*}
\int_{\Omega_t}|\bm{v}(\bm{x})|^2\,d\bm{x}
=\int_{\Omega_0}|\bm{V}(\bm{X})|^2\,J(\bm{X})\,d\bm{X}
\leq \|J\|_{L^{\infty}(\Omega_0)}\int_{\Omega_0}|\bm{V}(\bm{X})|^2\,d\bm{X},
\end{equation*}
hence
\begin{equation*}
\|\bm{v}\|^2_{L^2(\Omega_t)}\leq \|J\|_{L^{\infty}(\Omega_0)}\;\|\bm{V}\|^2_{L^2(\Omega_0)}.
\end{equation*}
Thus, $\bm{v} \in L^2(\Omega_t,\mathbb{R}^3)$.
\smallskip

\noindent
\underline{\emph{Step 1: first derivative.}}
Since $\bm{v} = \bm{V}\circ\pl^{-1}$, the chain rule yields
\begin{equation}
\label{first derivative vel}
\left[\nabla_{\bm{x}}\bm{v}(\bm{x})\right]_{ij} = \frac{\partial v_i}{\partial x_j} = \frac{\partial V_i}{\partial X_k}(\pl^{-1}(\bm{x}))\,\frac{\partial X_k}{\partial x_j} = \left[\nabla_{\bm{X}}\bm{V} (\pl^{-1}(\bm{x}))\right]_{ik}\left[\nabla_{\bm{x}}\pl^{-1}\right]_{kj}
\end{equation}
Equation \eqref{first derivative vel}, combined with \eqref{gradient inverse}, leads to
\begin{equation}
\label{change variable vel}
\nabla_{\bm{x}}\bm{v}(\pl(\bm{X}))
= \nabla_{\bm{X}}\bm{V}(\bm{X})\;\F^{-1}(\bm{X}).
\end{equation}
Exploiting the previous identity and a change of variables, recalling the fact that $J$ and $\F^{-1}$ are bounded in $\Omega_0$, we obtain
\begin{equation*}
\begin{aligned}
\|\nabla_{\bm{x}}\bm{v}\|^2_{L^2(\Omega_t)} &= \int_{\Omega_t}|\nabla_{\bm{x}}\bm{v}(\bm{x})|^2 d\bm{x} = \int_{\Omega_0}J(\bm{X})|\nabla_{\bm{x}}\bm{v}(\pl(\bm{X}))|^2d\bm{X}\\
&\leq\;\|J\|_{L^{\infty}(\Omega_0)}\;\|\F^{-1}\|_{L^\infty(\Omega_0)}^2\|\nabla_{\bm{X}}\bm{V}\|^2_{L^2(\Omega_0)},
\end{aligned}
\end{equation*}
hence we can conclude that $\nabla_{\bm{x}}\bm{v}\in L^2(\Omega_t,\mathbb{R}^9)$.

\smallskip

\noindent
\underline{\emph{Step 2: second derivative.}}
We start from the identity \eqref{change variable vel}, i.e., in components,
\begin{equation}
\label{first derivative_comp vel}
\frac{\partial V_i}{\partial X_a}(\pl^{-1}(\bm{x}))\,\frac{\partial X_a}{\partial x_j}(\bm{x}) = [\nabla_{\bm{x}}\bm{v}(\bm{x})]_{ij},
\qquad i,j=1,2,3.
\end{equation}
Differentiating \eqref{first derivative_comp vel} with respect to $x_k$ and applying the chain rule, one gets
\begin{equation}
\label{second derivative_comp vel}
\frac{\partial^2 V_i}{\partial X_b \partial X_a}(\pl^{-1}(\bm{x}))\frac{\partial X_b}{\partial x_k}(\bm{x})\frac{\partial X_a}{\partial x_j}(\bm{x})
\;+\;
\frac{\partial V_i}{\partial X_a}(\pl^{-1}(\bm{x}))\,\frac{\partial^2 X_a}{\partial x_k \partial x_j}(\bm{x})
=[\nabla^2_{\bm{x}}\bm{v}]_{ijk}.
\end{equation}
Let $\mathsf{A}$ and $\mathcal{H}$ be defined by the components
\begin{align*}
A_{a j}(\bm{X})&:=(\F^{-1})_{a j}(\bm{X})=\frac{\partial X_a}{\partial x_j}(\pl(\bm{X})),\\
H_{a k j}(\bm{X})&:=\left[\nabla_{\bm{x}}^2\pl^{-1}(\pl(\bm{X}))\right]_{a k j}=\frac{\partial^2 X_a}{\partial x_k\,\partial x_j}(\pl(\bm{X})).
\end{align*}
Equation \eqref{second derivative_comp vel} becomes
\begin{equation}
\begin{aligned}
\label{second derivative vel}
[\nabla^2_{\bm{x}}\bm{v}(\pl(\bm{X}))]_{ijk} &= [\nabla^2_{\bm{X}}\bm{V}(\bm{X})]_{iba}\,A_{b k}(\bm{X})A_{a j}(\bm{X})+
[\nabla_{\bm{X}}\bm{V}(\bm{X})]_{ia}\,H_{a k j}(\bm{X}).
\end{aligned}
\end{equation}
From \eqref{second derivative vel} and the submultiplicativity of the Frobenius norm, one has
\begin{equation}
\label{submultipl vel}
\big|\nabla^2_{\bm{x}}\bm{v}(\pl(\bm{X}))\big|
\;\le\;|\nabla^2_{\bm{X}}\bm{V}(\bm{X})|\,|\A(\bm{X})|^2\,+|\nabla_{\bm{X}}\bm{V}(\bm{X})|\,|\mathcal{H}(\bm{X})|.
\end{equation}
In particular, exploiting the previous pointwise estimate and a change of variables, one gets
\begin{equation*}
\begin{aligned}
&\|\nabla_{\bm{x}}^2\bm{v}\|_{L^2(\Omega_t)}^2
=
\int_{\Omega_t}\big|\nabla_{\bm{x}}^2\bm{v}(\bm{x})\big|^2\,d\bm{x}
=
\int_{\Omega_0}J(\bm{X})\big|\nabla_{\bm{x}}^2\bm{v}(\pl(\bm{X}))\big|^2\,d\bm{X}\\
&\le 2\int_{\Omega_0}J|\F^{-1}|^4\,|\nabla^2_{\bm{X}}\bm{V}|^2\,d\bm{X} + 2\int_{\Omega_0}J\,|\nabla_{\bm{X}}\bm{V}|^2\,|\nabla_{\bm{x}}^2\pl^{-1}(\pl(\bm{X}))|^2\,d\bm{X},
\end{aligned}
\end{equation*}
where we used also the fact that, for any $a$, $b$ real numbers, we have the inequality $(a+b)^2\leq 2a^2+2b^2$. As previously observed, both $J$ and $\F^{-1}$ are bounded in $\Omega_0$. Also $\nabla_{\bm{X}}\bm{V}$ is bounded in $\Omega_0$, since $\nabla_{\bm{X}}\bm{V}\in H^{s-1}(\Omega_0,\mathbb{R}^9)\hookrightarrow C^0(\overline{\Omega}_0, \mathbb{R}^9)$ and $\overline{\Omega}_0$ is compact. Therefore, we can conclude that
\begin{equation*}
\begin{aligned}
\|\nabla_{\bm{x}}^2\bm{v}\|^2_{L^2(\Omega_t)}
\;&\le\;2\|J\|_{L^{\infty}(\Omega_0)}\;\|\F^{-1}\|_{L^\infty(\Omega_0)}^4\,\|\nabla^2_{\bm{X}}\bm{V}\|^2_{L^2(\Omega_0)}\\
&+ 2\|J\|_{L^{\infty}(\Omega_0)}\;\|\nabla_{\bm{X}}\bm{V}\|_{L^\infty(\Omega_0)}^2\,\|\nabla_{\bm{x}}^2\pl^{-1}(\pl(\bm{X}))\|^2_{L^2(\Omega_0)},
\end{aligned}
\end{equation*}
hence $\nabla_{\bm{x}}^2\bm{v}\in L^2(\Omega_t,\mathbb{R}^{27})$ because $\nabla_{\bm{x}}^2\pl^{-1}(\pl(\bm{X}))\in L^2(\Omega_0,\mathbb{R}^{27})$ as shown in the proof of Theorem \ref{regularity inverse}.

\smallskip

\noindent
\underline{\emph{Step 3: third derivative.}}
We set
\[
G_{ijk}(\bm{X})=[\nabla^2_{\bm{x}}\bm{v}(\pl(\bm{X}))]_{ijk},
\]
that is
\begin{equation}
\label{Hess_vel_comp}
G_{ijk}(\bm{X}) = [\nabla^2_{\bm{X}}\bm{V}(\bm{X})]_{iba}\,A_{b k}(\bm{X})A_{a j}(\bm{X})+
[\nabla_{\bm{X}}\bm{V}(\bm{X})]_{ia}\,H_{a k j}(\bm{X}).
\end{equation}
By differentiating \eqref{Hess_vel_comp} with respect to $X_p$, one gets
\begin{equation*}
\begin{aligned}
\label{eq:dG_expanded}
&\frac{\partial G_{ijk}}{\partial X_p}(\bm{X})
=
\frac{\partial}{\partial X_p}[\nabla^2_{\bm{X}}\bm{V}]_{iba}\,A_{b k}\,A_{a j}
+[\nabla^2_{\bm{X}}\bm{V}]_{iba}\,\frac{\partial A_{b k}}{\partial X_p}\,A_{a j}\\
&+[\nabla^2_{\bm{X}}\bm{V}]_{iba}\,A_{b k}\frac{\partial A_{a j}}{\partial X_p}+\frac{\partial}{\partial X_p}[\nabla_{\bm{X}}\bm{V}]_{ia}\,H_{a k j} + [\nabla_{\bm{X}}\bm{V}]_{ia}\frac{\partial H_{a k j}}{\partial X_p}.
\end{aligned}
\end{equation*}
Taking into account that
\begin{equation*}
\label{isolated third gradient vel}
\left[\nabla_{\bm{x}}^3\bm{v}(\pl(\bm{X}))\right]_{ijk r}=\frac{\partial G_{ijk}}{\partial X_p}(\bm{X})\,A_{pr}(\bm{X}),
\end{equation*}
we arrive at the expression 
\begin{equation}
\label{third gradient vel}
\begin{aligned}
&\left[\nabla_{\bm{x}}^3\bm{v}(\pl(\bm{X}))\right]_{ijk r}=
[\nabla^3_{\bm{X}}\bm{V} (\bm{X})]_{iba p}\,A(\bm{X})_{bk}\,A(\bm{X})_{a j}\,A(\bm{X})_{pr}\\
&+[\nabla^2_{\bm{X}}\bm{V} (\bm{X})]_{ib a}\,H(\bm{X})_{bkr}\,A(\bm{X})_{a j}+[\nabla^2_{\bm{X}}\bm{V} (\bm{X})]_{ib a}\,H(\bm{X})_{ajr}\,A(\bm{X})_{b k}\\
&+[\nabla^2_{\bm{X}}\bm{V} (\bm{X})]_{ia p}\,H(\bm{X})_{a k j}\,A(\bm{X})_{pr}+[\nabla_{\bm{X}}\bm{V} (\bm{X})]_{ia}\,K(\bm{X})_{a k j r},
\end{aligned}
\end{equation}
where we used the notation
\[
K(\bm{X})_{ak j r} := \left[\nabla_{\bm{x}}^3\pl^{-1}(\pl(\bm{X}))\right]_{a k j r}.
\]
In order to find an explicit pointwise bound for $|\nabla_{\bm{x}}^3\bm{v}(\pl(\bm{X}))|$, we estimate each term in \eqref{third gradient vel} and exploit the submultiplicative property of the Frobenius norm. One then gets the following pointwise bound
\begin{equation}
\begin{aligned}
\label{pointwise_third vel}
\big|\nabla_{\bm{x}}^3\bm{v}(\pl(\bm{X}))\big|
&\le
|\A(\bm{X})|^3\,|\nabla^3_{\bm{X}}\bm{V}(\bm{X})|
+
3\,|\A(\bm{X})|\,|\nabla^2_{\bm{X}}\bm{V}(\bm{X})|\,|\mathcal{H}(\bm{X})|\\
&+ |\mathbb{K}(\bm{X})|\,|\nabla_{\bm{X}}\bm{V}(\bm{X})|.
\end{aligned}
\end{equation}
From the previous estimate, we obtain
\begin{equation*}
\begin{aligned}
\|\nabla_{\bm{x}}^3\bm{v}\|_{L^2(\Omega_t)}^2 &=
\int_{\Omega_0}J(\bm{X})\big|\nabla_{\bm{x}}^3\bm{v}(\pl(\bm{X}))\big|^2\,d\bm{X} \le C\left(\int_{\Omega_0}J|\F^{-1}|^6\,|\nabla_{\bm{X}}^3\bm{V}|^2\,d\bm{X}\right.\\
&+
\left.\int_{\Omega_0}J|\F^{-1}|^2\,|\nabla^2_{\bm{X}}\bm{V}|^2\,|\nabla_{\bm{x}}^2\pl^{-1}(\pl(\bm{X}))|^2\,d\bm{X}\right.\\
&\left.+
\int_{\Omega_0}J|\nabla_{\bm{x}}^3\pl^{-1}(\pl(\bm{X}))|^2\,|\nabla_{\bm{X}}\bm{V}|^2\,d\bm{X}\right),
\end{aligned}
\end{equation*}
where $C>0$ is a constant.

In particular, since $J$, $\F^{-1}$ and $\nabla_{\bm{X}}\bm{V}$ are bounded in $\Omega_0$
\begin{equation}
\begin{aligned}
\label{estimate third derivative vel}
\|\nabla_{\bm{x}}^3\bm{v}\|_{L^2(\Omega_t)}^2 &\leq  C\left(\;\|J\|_{L^{\infty}(\Omega_0)}\;\|\F^{-1}\|_{L^\infty(\Omega_0)}^6\,\|\nabla_{\bm{X}}^3\bm{V}\|_{L^2(\Omega_0)}^2\right.\\
&\left.+\;\|J\|_{L^{\infty}(\Omega_0)}\;\|\F^{-1}\|_{L^\infty(\Omega_0)}^{2}\,\|\left(\nabla^2_{\bm{X}}\bm{V}\right)\left(\nabla_{\bm{x}}^2\pl^{-1}(\pl(\bm{X}))\right)\|_{L^2(\Omega_0)}^2\right.\\
&\left. + \;\|J\|_{L^{\infty}(\Omega_0)}\;\|\nabla_{\bm{X}}\bm{V}\|_{L^\infty(\Omega_0)}^{2}\,\|\nabla_{\bm{x}}^3\pl^{-1}(\pl(\bm{X}))\|_{L^2(\Omega_0)}^2\right).
\end{aligned}
\end{equation}
The first term and the third term on the right-hand side of \eqref{estimate third derivative vel} are bounded, since $\nabla^3_{\bm{X}}\bm{V} \in L^2(\Omega_0,\mathbb{R}^{81})$ because $\bm{V}\in H^s(\Omega_0,\mathbb{R}^3)$ and $\nabla_{\bm{x}}^3\pl^{-1}(\pl(\bm{X}))\in L^2(\Omega_0,\mathbb{R}^{81})$ as shown in the proof of Theorem \ref{regularity inverse}. As for the second term on the right-hand side of \eqref{estimate third derivative vel}, recalling \eqref{submultipl}, we have
\begin{equation*}
\|\left(\nabla^2_{\bm{X}}\bm{V}\right)\left(\nabla_{\bm{x}}^2\pl^{-1}(\pl(\bm{X}))\right)\|_{L^2(\Omega_0)}^2 \leq \|\F^{-1}\|^6_{L^\infty(\Omega_0)}\|\left(\nabla^2_{\bm{X}}\bm{V}\right)\left(\nabla_{\bm{X}}\F\right)\|_{L^2(\Omega_0)}^2.
\end{equation*}
By H\"older's inequality, we also get
\[
\|\left(\nabla^2_{\bm{X}}\bm{V}\right)\left(\nabla_{\bm{X}}\F\right)\|_{L^2(\Omega_0)}^2
\le \|\nabla^2_{\bm{X}}\bm{V}\|^2_{L^6(\Omega_0)}\,\|\nabla_{\bm{X}}\F\|^2_{L^3(\Omega_0)}.
\]
Moreover, note that
\[
\nabla^2_{\bm{X}}\bm{V}, \nabla_{\bm{X}}\F\in H^{s-2}(\Omega_0,\mathbb{R}^{27})
\]
and $s-2\geq 1$, as $s\geq 3$. Hence
\[
\nabla^2_{\bm{X}}\bm{V}, \nabla_{\bm{X}}\F\in H^{1}(\Omega_0,\mathbb{R}^{27}).
\]
By the Sobolev embedding theorem in dimension three,
\[
H^1(\Omega_0, \mathbb{R}^{27}) \hookrightarrow L^6(\Omega_0,\mathbb{R}^{27}),
\]
and since $\Omega_0$ is bounded, the inclusion $L^6(\Omega_0, \mathbb{R}^{27}) \subseteq L^3(\Omega_0, \mathbb{R}^{27})$
also holds. Consequently, also the second term in \eqref{estimate third derivative vel} is bounded. Thus, $\nabla^3_{\bm{x}}\bm{v}\in L^2(\Omega_t,\mathbb{R}^{81})$.

\smallskip

\noindent
\underline{\emph{Step 4: higher-order derivatives.}}
For $s>3$, define
\[
\mathbb{L}_s(\bm{X}) := \nabla_x^{\,s} \bm{v}(\pl(\bm{X})).
\]
Note that one has the following estimate on the Frobenius norm of $\mathbb{L}_s(\bm{X})$:
\begin{equation}
\label{estimate gradient s vel}
|\mathbb{L}_s(\bm{X})| \le C_s \sum_{m=1}^{s} |\nabla_{\bm{X}}^m\bm{V}(\bm{X})| \sum_{\substack{k_1,\dots,k_{s-1}\ge 0\\ \sum_{j=1}^{s-1}j\,k_j=s-m}} |\A(\bm{X})|^{\,s+K} \prod_{j=1}^{s-1}|\nabla_{\bm{X}}^j\F(\bm{X})|^{k_j},
\end{equation}
where $K := \sum_{j=1}^{s-1} k_j$ and $C_s>0$ is a constant. Recalling also \eqref{estimate gradient s}, for $s=1,2,3$, the inequality \eqref{estimate gradient s vel} corresponds exactly to \eqref{change variable vel}, \eqref{submultipl vel} and \eqref{pointwise_third vel}. 
We can now prove the statement by induction with a simple adaptation of the argument given in the proof of Theorem \ref{regularity inverse}.
\end{proof}

\section{Geodesic flow and local well-posedness}\label{sec:geodesics}

In this section, we explore the geometric framework underlying the inertial motion of incompressible continua and establish the local well-posedness of the associated evolution problem. Building on the variational formulation obtained in the previous section, we show that the equations of motion can be interpreted as geodesic equations on the submanifold of volume-preserving deformations. The geometric structure obtained here extends the Arnold--Ebin--Marsden construction (see \cite{ebin1970groups}, \cite{arnold1966geometrie}) to a setting that includes general volume-preserving deformations, not limited to fluid motions within a fixed domain.

\begin{lemma}
\label{lem:submanifold}
Let \(\Omega_0 \subset \mathbb{R}^3\) be a bounded connected domain with Lipschitz boundary and let $s\geq 3$ be an integer. Define the Jacobian determinant map
\[
J:\DSP\longrightarrow H^{s-1}(\Omega_0,\mathbb{R}),
\]
given by $J(\pl)(\bm{X})= \operatorname{det}\F(\bm{X})$, for every $\bm{X}\in\Omega_0$. Then $J$ is a submersion on
\[
\DSM := \{ \pl \in \DSP \mid J(\pl) = 1 \}.
\]
\end{lemma}

\begin{proof}
We first compute the differential (tangent map) of $J$ and then prove surjectivity at each point in $\DSM$ in order to state that $J$ is a submersion on $\DSM$ (see \cite{lang2002introduction}).
For any $\pl \in \DSP$, the differential of $J$ at $\pl$ is
\[
T_{\pl}J : T_{\pl}\DSP\to H^{s-1}(\Omega_0,\mathbb{R})
\]
given by
\begin{equation*}
T_{\pl}J(\bm{V}) = J(\pl) (\operatorname{div}_{\bm{x}}\bm{v})\circ\pl,
\end{equation*}
where $\bm{V}=\bm{v}\circ\pl\in H^{s}(\Omega_0,\mathbb{R}^3)$. This result follows from Proposition \ref{tangent vectors} and the standard formula for the differential of the determinant (see \cite{marsden1994mathematical}, \cite{gurtin1981introduction})
\[
T_{\pl}J(\bm{V}) = J(\pl)\,\mathrm{tr}(\F^{-1}\nabla_{\bm{X}} \bm{V}) = J(\pl) \F^{-\top} : \nabla_{\bm{X}}\bm{V} = J(\pl) (\operatorname{div}_{\bm{x}}\bm{v})\circ\pl,
\]
where we used the matrix product $\A:\mathsf{B} = \operatorname{tr}(\A\mathsf{B}^\top)$. Moreover, restricting our attention to the set $\DSM$, we have $J(\pl)=1$ in the above formulae. Therefore, we need to prove that for every $g\in H^{s-1}(\Omega_0,\mathbb{R})$, there exists $\bm{V}\in H^{s}(\Omega_0,\mathbb{R}^3)$ such that
\begin{equation}
\label{surjectivity2}
\F^{-\top} : \nabla_{\bm{X}}\bm{V} = g.
\end{equation}
The change of variable $\bm{X} = \pl^{-1}(\bm{x})$ yields
\begin{equation}
\label{surj spatial2}
\operatorname{div}_{\bm{x}}\bm{v} = \hat{g},
\end{equation}
where $\bm{v}\in H^s(\Omega,\mathbb{R}^3)$ thanks to Proposition \ref{prop:regularity velocity} and $\hat{g} = g\circ\pl^{-1}\in H^{s-1}(\Omega,\mathbb{R})$ according to Remark \ref{composition diffeo}. Recalling that $\Omega_0\subset\mathbb{R}^3$ is a bounded domain with Lipschitz boundary, we have that $\Omega=\pl(\Omega_0)$ is also a bounded Lipschitz domain in $\mathbb{R}^3$, due to the regularity of
\[
\pl\in\DSP\hookrightarrow C^1(\overline{\Omega}_0,\mathbb{R}^3).
\]
In order to prove the existence of a solution for equation \eqref{surj spatial2}, we use the regularity properties of the de Rham complex without boundary conditions on starlike domains with respect to a ball. More precisely, exploiting the properties of Poincaré-type operators (see for example \cite{BottTu82} for a standard proof of Poincaré's lemma), it is possible to prove that, for $\Omega\subset\mathbb{R}^n$ a bounded Lipschitz domain (and therefore the finite union of starlike domains with respect to some ball, see for example \cite{galdi2011introduction}), the complex
\begin{equation}
\label{deRham}
0 \longrightarrow H^{\sigma}(\Omega,\Lambda^{0})
\xrightarrow{\,d\,}
H^{\sigma-1}(\Omega,\Lambda^{1})
\xrightarrow{\,d\,}\cdots\xrightarrow{\,d\,}
H^{\sigma-n}(\Omega,\Lambda^{n})
\longrightarrow 0,
\end{equation}
has a finite-dimensional cohomology space for any $\sigma\in\mathbb{R}$ (see \cite{constabel2010}). Recall that in \eqref{deRham} $d$ is the exterior derivative satisfying $d\circ d=0$ and $\Lambda^{\ell}$, with $0\leq\ell\leq n$, denotes the exterior algebra of $\mathbb{R}^n$. In particular, for $n=3$, the complex \eqref{deRham} can be written as
\[
0 \longrightarrow H^{\sigma+3}(\Omega,\Lambda^{0})
\xrightarrow{\,\operatorname{grad}\,}
H^{\sigma+2}(\Omega,\Lambda^{1})
\xrightarrow{\,\operatorname{curl}\,}
H^{\sigma+1}(\Omega,\Lambda^{2})
\xrightarrow{\,\operatorname{div}\,}
H^{\sigma}(\Omega,\Lambda^{3})
\longrightarrow 0,
\]
where we can identify $\Lambda^0$, $\Lambda^3$ with $\mathbb{R}$ and $\Lambda^1$, $\Lambda^2$ with $\mathbb{R}^3$. Moreover, for any $\sigma \in \mathbb{R}$, we have that in top degree the cohomology vanishes (see \cite{constabel2010}):
\begin{equation}
\label{0 cohomology}
d\,H^{\sigma+1}(\Omega,\Lambda^{2})=H^{\sigma}(\Omega,\Lambda^{3}).
\end{equation}
We can therefore apply this result for $\sigma=s-1$. From \eqref{0 cohomology}, it follows that, given $\hat{g}\in H^{s-1}(\Omega,\mathbb{R})$, there exists $\bm{v}\in H^s(\Omega,\mathbb{R}^3)$ such that
\[
\operatorname{div}_{\bm{x}}\bm{v}=\hat{g}.
\]
Therefore we have that $\bm{V} = \bm{v}\circ \pl \in H^s(\Omega_0,\mathbb{R}^3)$ by Proposition \ref{prop:regularity velocity} and Remark \ref{composition diffeo} and solves
\begin{equation*}
\F^{-\top} : \nabla_{\bm{X}}\bm{V} = g,
\end{equation*}
where $g = \hat{g}\,\circ\,\pl\in H^{s-1}(\Omega_0,\mathbb{R})$. Thus we have shown that $\bm{V}$ belongs to $H^s(\Omega_0,\mathbb{R}^3)$ and solves \eqref{surjectivity2}. Hence, $T_{\pl} J$ is surjective for every $\pl\in D^s_{\mu}$, which
proves the claim.
\end{proof}

\begin{remark}
Note that $J$ takes values in $H^{s-1}(\Omega_0,\mathbb{R})$ since the determinant is a polynomial function of the entries of $\F\in H^{s-1}(\Omega_0,\operatorname{Mat}_{3}(\mathbb{R}))$.
Indeed, for $s > 5/2$, the Sobolev space $H^{s-1}(\Omega_0,\mathbb{R})$ is an algebra, i.e., if $f,g \in H^{s-1}(\Omega_0,\mathbb{R})$, then $fg \in H^{s-1}(\Omega_0,\mathbb{R})$ (see for example \cite{taylor2011partial}).
\end{remark}

The previous lemma allows us to easily prove the following result, establishing the submanifold structure of $\DSM$.
\begin{theorem}
Let $\Omega_0 \subset \mathbb{R}^3$ be a bounded connected domain with Lipschitz boundary and let $s \geq 3$ be an integer. Let $\DSM$ be the set of volume-preserving maps in $\DSP$, namely
\[
\DSM := \{ \pl \in \DSP \mid J(\pl) = 1 \}.
\]
Then $\DSM$ is a submanifold of $\DSP$, with tangent space at $\pl$ given by
\begin{equation}
\label{tangent space}
T_{\pl}\DSM=\{\bm{V}\in H^s(\Omega_0,\mathbb{R}^3)\mid \bm{V}=\bm{v}\circ\pl,~\bm{v}\in H^s(\Omega, \mathbb{R}^3) \text{ with }\operatorname{div}_{\bm{x}}\bm{v} = 0\}.
\end{equation}
\end{theorem}
\begin{proof}
Lemma \ref{lem:submanifold} showed that $J:\DSP\longrightarrow H^{s-1}(\Omega_0,\mathbb{R})$ is a submersion on $\DSM$. Hence, it follows immediately that
\[
\DSM = J^{-1}(1)
\]
is a submanifold of $\DSP$ (see \cite{lee2013smooth}). Moreover, the tangent space to $\DSM$ is the kernel of the tangent map $T_{\pl}J : T_{\pl}\DSP\to H^{s-1}(\Omega_0,\mathbb{R})$, that gives immediately \eqref{tangent space}. 
\end{proof}

The inertial motion considered in equation \eqref{Lagrangian Euler equation} can be equivalently described as the geodesic flow on the submanifold $\DSM\subset \DSP$ of volume-preserving
deformations, endowed with the natural kinetic energy (Riemannian) metric. More precisely, the
kinetic energy functional \eqref{kinetic energy} induces a (weak) Riemannian structure on $\DSP$, with (weak) Riemannian metric
\begin{equation}
\label{metric}
\langle \bm{U}, \bm{W}\rangle_{\pl}=\int_{\Omega_0}\rho_0\;\bm{U}(\bm{X})\cdot \bm{W}(\bm{X})\,d\bm{X},
\end{equation}
where $\bm{U},\bm{W}\in T_{\pl} \DSP$. By definition, a geodesic on the submanifold $\DSM\subset \DSP$ of volume-preserving
deformations with respect to the metric defined in \eqref{metric} is a stationary curve $\pl: [0,T] \to \DSM$ for the functional
\[
\mathcal{F}[\pl] := \frac{1}{2}\int_0^T \langle \dot{\pl}(t), \dot{\pl}(t)\rangle_{\pl(t)}dt = \int_0^T\int_{\Omega_0}\frac{1}{2}\rho_0\rvert\dot{\pl}(\bm{X},t)\lvert^2d\bm{X}dt.
\]
Thus, finding a geodesic on $\DSM$ is equivalent to finding a curve in $\DSM$ that minimizes the kinetic energy functional under the incompressibility constraint, i.e., finding a solution for the equation of inertial motion \eqref{Lagrangian Euler equation}.

Having identified inertial motions with geodesics on $\DSM$, we now wish to discuss the local well-posedness of this geodesic flow, in the spirit of \cite{ebin1970groups}. In doing so, we are inspired by the works of Ebin~\cite{ebin2015groups} and Inci~\cite{inci2015lagrangian,inci2015regularity} on the Lagrangian formulation of the Euler equation.

Following the construction of \cite{ebin2015groups}, and recalling the arguments of Proposition \ref{ortho}, we denote by $P$ and $Q$ the spatial projectors onto the orthogonal subspaces of divergence-free vector spaces and gradients of scalar fields that vanish on the boundary, respectively. Therefore, given a vector field $\bm{u}=\bm{v}+\nabla_{\bm{x}}p$ on $\Omega_t$, we have that $P(\bm{u}) = \bm{v}$, with $\operatorname{div}_{\bm{x}}\bm{v}=0$, and $Q(\bm{u}) = \nabla_{\bm{x}}p$, with $p=0$ on $\partial\Omega_t$. Recalling that
\begin{equation*}
\label{material div}
\operatorname{div}_{\bm{x}}\bm{v}=\operatorname{tr}\left(\F^{-1}\nabla_{\bm{X}}\mathbf{V}\right) = \F^{-\top} \colon \nabla_{\bm{X}}\bm{V},
\end{equation*}
where $\bm{V}=\bm{v}\,\circ\,\pl$ with $\pl\in\DSM$, we can define the Lagrangian projectors $\widehat{P}$ and $\widehat{Q}$ as follows. Given a vector field $\bm{U}$ on $\Omega_0$, $\widehat{P}(\bm{U}) := \bm{V}$, with $\F^{-\top} \colon \nabla_{\bm{X}}\bm{V}=0$, and $\widehat{Q}(\bm{U}) := \F^{-\top}\nabla_{\bm{X}}\hat{p}$, with $\hat{p}=0$ on $\partial\Omega_0$. Since
\begin{equation*}
\begin{split}
\ddot{\pl} &= \dot{\bm{V}}(\bm{X},t) = \frac{\partial\bm{v}}{\partial t}(\pl(\bm{X},t),t)+\nabla_{\bm{x}}\bm{v}(\pl(\bm{X},t),t)\dot{\pl}(\bm{X},t)\\
&=\frac{\partial\bm{v}}{\partial t}(\pl(\bm{X},t),t)+\dot{\F}\F^{-1}\dot{\pl}(\bm{X},t) = \frac{\partial\bm{v}}{\partial t}(\pl(\bm{X},t),t)+(\nabla_{\bm{X}}\bm{V}(\bm{X},t))\F^{-1}\bm{V}(\bm{X},t),
\end{split}
\end{equation*}
equation \eqref{Lagrangian Euler equation} can also be viewed as
\begin{equation*}
\widehat{P}\left(\rho_0\ddot{\pl} \right) = \rho_0\widehat{P}\left(\frac{\partial\bm{v}}{\partial t}\circ\pl+(\nabla_{\bm{X}}\bm{V})\F^{-1}\dot{\pl}\right) = -\widehat{P}\left(\F^{-\top}\nabla_{\bm{X}}\hat{p}\right)=0.
\end{equation*}
Moreover, since $\bm{v}$, and therefore $\frac{\partial\bm{v}}{\partial t}$, are divergence-free, $\widehat{Q}\left(\frac{\partial\bm{v}}{\partial t}\circ\pl\right) =0$. Thus we get
\begin{equation*}
\frac{\partial\bm{v}}{\partial t}\circ\pl + \widehat{P}\left((\nabla_{\bm{X}}\bm{V})\F^{-1}\dot{\pl}\right) = 0,
\end{equation*}
or equivalently
\begin{equation}
\label{eq projector Q}
\frac{\partial\bm{v}}{\partial t}\circ\pl +\left(\nabla_{\bm{X}}\bm{V}\right)\F^{-1}\dot{\pl} = \widehat{Q}\left((\nabla_{\bm{X}}\bm{V})\F^{-1}\dot{\pl}\right).
\end{equation}
We now want to show that equation \eqref{eq projector Q} can be construed as a second-order ODE on $\DSM$, in order to exploit the fundamental theory of ODEs to prove its well-posedness.
\begin{proposition}
\label{prop:ODE}
Let $\Omega_0 \subset \mathbb{R}^3$ be a bounded connected domain with Lipschitz boundary and let $s \geq 3$ be an integer. The equation of motion \eqref{Lagrangian Euler equation} for an incompressible homogeneous body is a second-order ODE on $\DSM$ of the form $\ddot{\pl}= Z(\pl,\dot{\pl})$.
\end{proposition}
\begin{proof}
We recall the decomposition
\begin{equation*}
\bm{U}(\bm{X})= \bm{V}(\bm{X})+\F^{-\top}\nabla_{\bm{X}}\hat{p}(\bm{X}),
\end{equation*}
with $\F^{-\top}\colon \nabla_{\bm{X}}\bm{V}(\bm{X})=0$ in $\Omega_0$ and $\hat{p}(\bm{X})=0$ on $\partial\Omega_0$. If we differentiate both sides of the equation above and we multiply them by $\F^{-\top}$, we get
\begin{equation}
\begin{aligned}
\label{lagrangian laplacian pressure}
&\F^{-\top}\colon \nabla_{\bm{X}}\bm{U}(\bm{X}) = \F^{-\top}\colon \nabla_{\bm{X}}\left(\F^{-\top}\nabla_{\bm{X}}\hat{p}(\bm{X})\right)\\
&= \operatorname{div}_{\bm{x}}(\nabla_{\bm{x}}p(\bm{x}))\rvert_{\bm{x}=\pl(\bm{X})} =\Delta_{\bm{x}}p(\bm{x})\rvert_{\bm{x}=\pl(\bm{X})} =: \widetilde{\Delta}_{\bm{X}}\hat{p}(\bm{X}),
\end{aligned}
\end{equation}
with $p=\hat{p}\circ\pl^{-1}$.
Notice that $\widetilde{\Delta}_{\bm{X}}$ is invertible, since $\widetilde{\Delta}_{\bm{X}}^{-1}=\pl(\cdot,t)\circ \Delta_{\bm{x}}^{-1} \circ \pl^{-1}(\cdot,t)$.
Hence, as for a spatial vector field $\bm{w}$ we can write $Q(\bm{w}) = \nabla_{\bm{x}}\Delta_{\bm{x}}^{-1}\operatorname{div}\bm{w}$, for a material vector field $\bm{W}$ we can invert the operator $\widetilde{\Delta}_{\bm{X}}$ in equation \eqref{lagrangian laplacian pressure} and write
\begin{equation}
\label{eq Q Lagr}
\widehat{Q}(\bm{W}) = \F^{-\top}\nabla_{\bm{X}}\widetilde{\Delta}_{\bm{X}}^{-1}\left(\F^{-\top}\colon\nabla_{\bm{X}}\bm{W}\right).
\end{equation}
By choosing $\bm{W}=(\nabla_{\bm{X}}\bm{V})\F^{-1}\dot{\pl}$ in \eqref{eq Q Lagr}, equation \eqref{eq projector Q}, which is a projection of \eqref{Lagrangian Euler equation}, becomes
\begin{equation}
\label{ODE Euler}
\ddot{\pl}(\bm{X},t) = \F^{-\top}\nabla_{\bm{X}}\widetilde{\Delta}_{\bm{X}}^{-1}\left(\F^{-\top}\colon\nabla_{\bm{X}}\left((\nabla_{\bm{X}}\bm{V})\F^{-1}\dot{\pl}(\bm{X},t)\right)\right).
\end{equation}
Moreover, one has
\begin{equation}
\label{order 1 operator}
\begin{split}
&\F^{-\top}\colon\nabla_{\bm{X}}\left((\nabla_{\bm{X}}\bm{V})\F^{-1}\dot{\pl}\right)
=(\F^{-\top})_{ij}\frac{\partial}{\partial X_j}\left[\left((\nabla_{\bm{X}}\bm{V})\F^{-1}\right)_{ik}V_k\right]\\
&=\frac{\partial X_j}{\partial x_i}\frac{\partial}{\partial X_j}\left((\nabla_{\bm{X}}\bm{V})\F^{-1}\right)_{ik}V_k + \frac{\partial X_j}{\partial x_i}\left((\nabla_{\bm{X}}\bm{V})\F^{-1}\right)_{ik}\frac{\partial V_k}{\partial X_j}.
\end{split}
\end{equation}
The first term in equation \eqref{order 1 operator} vanishes due to the incompressibility constraint, because
\begin{equation*}
\label{vanishing part}
\begin{split}
\frac{\partial X_j}{\partial x_i}\frac{\partial}{\partial X_j}\left((\nabla_{\bm{X}}\bm{V})\F^{-1}\right)_{ik} &= [\operatorname{div}_{\bm{x}}(\nabla_{\bm{x}}\bm{v}(\bm{x})^\top)]_k\rvert_{\bm{x}=\pl(\bm{X})}\\
&=[\nabla_{\bm{x}}\operatorname{div}_{\bm{x}}\bm{v}(\bm{x})]_k\rvert_{\bm{x}=\pl(\bm{X})} = 0.
\end{split}
\end{equation*}
This leaves us with
\begin{equation*}
    \F^{-\top}\colon\nabla_{\bm{X}}\left((\nabla_{\bm{X}}\bm{V})\F^{-1}\dot{\pl}\right) =\frac{\partial X_j}{\partial x_i}\left((\nabla_{\bm{X}}\bm{V})\F^{-1}\right)_{ik}\frac{\partial V_k}{\partial X_j}= \frac{\partial X_j}{\partial x_i}\frac{\partial V_i}{\partial X_\ell}\frac{\partial X_\ell}{\partial x_k}\frac{\partial V_k}{\partial X_j}.
\end{equation*}
The above expression is smooth in $(\pl,\bm{V})$ as a function from $H^s$ to $H^{s-1}$ since it involves first derivatives of $\bm{V}$ and since $\F^{-1}$, $\F^{-\top}$ depend smoothly on $\pl$. Thus, the operator
\begin{equation}
\label{operator Z}
Z(\pl(\cdot,t),\dot{\pl}(\cdot,t)) = \F^{-\top}\nabla_{\bm{X}}\widetilde{\Delta}_{\bm{X}}^{-1}\left(\F^{-\top}\colon\nabla_{\bm{X}}\left((\nabla_{\bm{X}}\bm{V})\F^{-1}\dot{\pl}(\bm{X},t)\right)\right)
\end{equation}
is smooth from $H^s$ to $H^s$, since $\widetilde\Delta_{\bm{X}}^{-1}$ is obtained by
conjugating the inverse spatial Laplacian by $\pl$ --- and therefore depends
smoothly on $\pl$ --- and since $\nabla_{\bm{X}}\widetilde{\Delta}_{\bm{X}}^{-1}$ is a pseudodifferential operator of order minus one (see for example \cite{Taylor81}). Therefore, \eqref{ODE Euler} is a smooth ODE on $\DSM$.
\end{proof}
\begin{theorem}
\label{thm:well posedness}
Let $\Omega_0 \subset \mathbb{R}^3$ be a bounded connected domain with Lipschitz boundary and let $s \geq 3$ be an integer. Then the Cauchy problem associated to the inertial motion \eqref{Lagrangian Euler equation} for an incompressible homogeneous continuum
\begin{equation}
\label{Cauchy problem1}
\begin{cases}
\ddot{\pl}(\cdot,t)= Z(\pl(\cdot,t),\dot{\pl}(\cdot,t))\\
\pl(\cdot,0) = \pl_0(\cdot)\\
\dot{\pl}(\cdot,0)=\bm{V}_0(\cdot),
\end{cases}
\end{equation}
with $Z$ defined as in \eqref{operator Z}, is locally well-posed for initial data $(\pl_0,\bm{V}_0)\in T\DSM$.
\end{theorem}
\begin{proof}
The result follows directly from Proposition \ref{prop:ODE}. In fact, since \eqref{Lagrangian Euler equation} is a smooth ODE on $\DSM$, the classical Cauchy--Lipschitz theorem for ODEs (see for example \cite{coddington1955theory}, \cite{Robbin1968}) ensures that, given the initial conditions
\[
\pl(\cdot,0) = \pl_0(\cdot) \in \DSM\quad\text{and}\quad\dot{\pl}(\cdot,0)=\bm{V}_0 \in T_{\pl_0}\DSM,
\]
there exists a unique maximal solution
\[
\pl \colon (-T_{-},T_{+}) \to \DSM,
\]
with $T_{-}$ and $T_{+}$ positive or infinite. Moreover, the solution depends smoothly on the initial data.
\end{proof}
Although the equation of inertial motion coincides, as underlined in Remark \ref{rmk:inertial motion}, with the physical equation of motion only for a very specific class of materials, incompressible
perfect fluids, its geometric significance is much broader. Indeed, the inertial equation defines the geodesic flow associated with the kinetic energy metric on the manifold
of admissible configurations, independently of whether the resulting evolution
represents the actual motion of a given material. More precisely, let $\mathcal{M}$ denote either the manifold $\DSP$ of admissible
deformations or its volume-preserving submanifold $\DSM$, and let
$\pl_0\in\mathcal{M}$. For every initial velocity
$\bm{V}_0\in T_{\pl_0}\mathcal{M}$, the well-posedness of the geodesic equation established in Theorem \ref{thm:well posedness} provides a unique maximal solution $\pl(t)$, with $\pl(0)=\pl_0$, and $\dot\pl(0)=\bm{V}_0$. This naturally gives rise to the exponential map
\[
\exp_{\pl_0}(\bm{V}_0)=\pl(1),
\]
whenever the solution exists up to time $t=1$. The exponential map suggests a natural parametrization of a neighborhood of the initial configuration by means of its tangent space. 

In this sense, inertial motions
may provide an intrinsic geometric description of the manifold of admissible
configurations. The domain of validity of this parametrization is
determined by the maximal interval of existence of the corresponding geodesic. In other words, the maximal interval of existence of the solution of the Cauchy problem \eqref{Cauchy problem1} determines the portion of configuration manifold that can be parametrized by following the geodesic starting from the initial deformation $\pl_0$. From this viewpoint, the equation of inertial motion is a fundamental geometric tool for exploring the
local structure of the configuration manifold for a generic continuum.

\section{Examples}\label{sec:examples}

In this section, we wish to highlight the difference between compressible and incompressible geodesic motions by means of a few examples.
Particular attention is devoted to the maximal interval on which these curves remain in the corresponding configuration manifold and to the mechanisms through which admissibility may be lost.
In the first example, an open ball is compressed to a single point, meaning that the geodesic reaches the boundary of the admissible configuration manifold, in finite time.
In the second and third examples we show that geodesic flows with identical initial conditions can exist for all times in the incompressible case, while exiting the manifold in finite time in the compressible case.

\begin{example}
\label{ex:isotropic contraction}
Let the material manifold $\Omega_0\subset\mathbb{R}^3$ be the open ball of radius $R>0$ centered at the origin. Consider the manifold $\DSP$ of orientation-preserving deformations, endowed with the kinetic energy metric \eqref{metric}.
In the compressible case, where no volume-preserving constraint is imposed, the geodesic equation reduces to
\[
\ddot{\pl}(\bm{X},t)=0
\]
for every $\bm{X}\in\Omega_0$. Let the initial configuration be the identity, namely $\pl_0(\bm{X})=\bm{X}$, and choose the initial Lagrangian velocity field as
\[
\bm{V}_0(\bm{X})=-\bm{X},
\]
so that every material point starts moving towards the origin, with a velocity whose magnitude is proportional to its distance from the origin.
The unique solution of the geodesic equation is therefore
\begin{equation}
\label{solution compressible example}
\pl(\bm{X},t)=\pl_0(\bm{X})+t\bm{V}_0(X)=(1-t)\bm{X}.
\end{equation}
The motion consists of a homogeneous isotropic contraction of the body. Indeed, every material point moves along the straight segment joining its initial
position to the origin, with a constant velocity equal to its initial velocity. The motion is purely inertial, since every particle has vanishing acceleration. The deformation gradient is
\[
\F(\bm{X},t)=\nabla_{\bm{X}}\pl(\bm{X},t)=(1-t)\mathsf{I},
\]
so that
\[
J(\bm{X},t)=\det \F(\bm{X},t)=(1-t)^3.
\]
Hence, for every $t<1$, one has $J(\bm{X},t)>0$ and therefore $\pl(\cdot,t)\in\DSP$. At the critical time $t=1$, one has $\pl(\bm{X},1)=\bm{0}$ for every $\bm{X}\in\Omega_0$, that is, every material point occupies the same spatial position, namely the origin. Notice that the image of the body remains a ball throughout the evolution. More
precisely, $\Omega_t=\pl_t(\Omega_0)$ is the ball of radius $(1-t)R$ centered at the origin, that is, the radius decreases linearly until the whole body collapses into a single point at $t=1$. Consequently, $J(\bm{X},1)=0$, so the deformation is no longer locally invertible.
The placement map ceases to be injective, violating the condition of non-interpenetration of matter, and the configuration no longer belongs to the manifold $\DSP$ of admissible deformations. Therefore, although the geodesic equation itself admits the global solution \eqref{solution compressible example} for every $t\in\mathbb R$, the corresponding geodesic on $\mathcal D^s_+$ exists only for $t<1$.
The maximal existence interval is thus determined not by a breakdown of the
ordinary differential equation, but by the loss of admissibility of the deformation: the geodesic reaches
the boundary of the manifold of admissible configurations, where the Jacobian
determinant vanishes.
\end{example}

\begin{example}
\label{ex:incompressible}
For simplicity, we restrict our attention to two dimensions. Let the material manifold $\Omega_0\subset\mathbb{R}^2$ be the open unit disk
\[
\Omega_0 = \left\{
        \bm{X}=(X,Y)\in\mathbb{R}^2
        :
        X^2+Y^2<1
    \right\}.
\]
Consider the manifold $\DSM$ of volume-preserving deformations, endowed with the kinetic energy metric \eqref{metric}. The geodesic equation on this manifold is given by \eqref{Lagrangian Euler equation} with boundary condition \eqref{boundary cond}. Let the initial configuration be the identity, namely
\begin{equation}
\label{initial deformation}
    \pl_0(\bm{X})=\bm{X},
\end{equation}
and choose the initial Lagrangian velocity field as
\begin{equation}
\label{initial velocity}
\bm{V}_0(\bm{X})=\begin{pmatrix}
X\\
-Y
\end{pmatrix},
\end{equation}
so that, at the initial time, the body is stretched in the horizontal
direction and compressed in the vertical direction.
The unique solution of the geodesic equation with these initial conditions is therefore given by
\begin{equation}
\label{eq:extensional-placement}
    \bm{\varphi}(\bm{X},t)
    =
    \begin{pmatrix}
        e^{\varepsilon(t)}X\\
        e^{-\varepsilon(t)}Y
    \end{pmatrix},
\end{equation}
where the scalar function $\varepsilon\colon\mathbb{R}\to\mathbb{R}$ satisfies
\begin{equation}
\label{eq:coefficient-equality}
    \left(\ddot{\varepsilon}(t)
          +\dot{\varepsilon}^2(t)\right)
    e^{2\varepsilon(t)}
    =
    \left(\dot{\varepsilon}^2(t)
          -\ddot{\varepsilon}(t)\right)
    e^{-2\varepsilon(t)},
\end{equation}
with $\varepsilon(0)=0$ and $\dot{\varepsilon}(0)=1$. In fact, \eqref{eq:extensional-placement} clearly satisfies the initial conditions \eqref{initial deformation} and \eqref{initial velocity}. Moreover, the inertial motion equation \eqref{Lagrangian Euler equation} equivalently reads
\begin{equation}
\label{eq:material-pressure-gradient}
    \nabla_{\bm{X}}\hat{p}
    =
    -\rho_0\F^\top\ddot{\bm{\varphi}}.
\end{equation}
Differentiating \eqref{eq:extensional-placement} with respect to time,
we obtain
\[
    \dot{\bm{\varphi}}(\bm{X},t)
    =
    \begin{pmatrix}
        \dot{\varepsilon}(t)e^{\varepsilon(t)}X\\
        -\dot{\varepsilon}(t)e^{-\varepsilon(t)}Y
    \end{pmatrix}.
\]
A second differentiation gives
\begin{equation}
\label{eq:extensional-acceleration}
    \ddot{\bm{\varphi}}(\bm{X},t)
    =
    \begin{pmatrix}
        \left(\ddot{\varepsilon}(t)
              +\dot{\varepsilon}^2(t)\right)
        e^{\varepsilon(t)}X\\
        \left(\dot{\varepsilon}^2(t)
              -\ddot{\varepsilon}(t)\right)
        e^{-\varepsilon(t)}Y
    \end{pmatrix}.
\end{equation}
The deformation gradient associated with \eqref{eq:extensional-placement} is
\begin{equation*}
\label{eq:extensional-F}
    \F(\bm{X},t) =
    \begin{pmatrix}
        e^{\varepsilon(t)} & 0\\
        0 & e^{-\varepsilon(t)}
    \end{pmatrix}.
\end{equation*}
It is independent of the material point $\bm{X}$ and satisfies
\[
J(\bm{X},t) = \det\F(\bm{X},t) = e^{\varepsilon(t)}e^{-\varepsilon(t)} = 1.
\]
Hence the motion is incompressible for every time for which
$\varepsilon(t)$ is defined. The current configuration is
\begin{equation*}
\label{eq:current-ellipse}
    \Omega_t
    =
    \bm{\varphi}_t(\Omega_0)
    =
    \left\{
        (x,y)\in\mathbb{R}^2
        :
        e^{-2\varepsilon(t)}x^2
        +
        e^{2\varepsilon(t)}y^2
        <1
    \right\}.
\end{equation*}
Thus the disk is transformed into an ellipse whose semi-axes are $R_x(t)=e^{\varepsilon(t)}$ and $R_y(t)=e^{-\varepsilon(t)}$. Their product is equal to one, consistently with the preservation of area. Since $\F^\top=\F$, from
\eqref{eq:material-pressure-gradient} and
\eqref{eq:extensional-acceleration} it follows that
\begin{equation}
\label{eq:pressure-gradient-example}
    \nabla_{\bm{X}}\hat{p}
    =
    -\rho_0
    \begin{pmatrix}
        \left(\ddot{\varepsilon}(t) +\dot{\varepsilon}^2(t)\right)
        e^{2\varepsilon(t)}X\\
        \left(\dot{\varepsilon}^2(t)
              -\ddot{\varepsilon}(t)\right)
        e^{-2\varepsilon(t)}Y
    \end{pmatrix}.
\end{equation}
Integrating \eqref{eq:pressure-gradient-example} with respect to the
material variables $X$ and $Y$, we find
\begin{equation}
\label{eq:material-pressure-general}
\hat{p}(\bm{X},t) =-\frac{\rho_0}{2}
    \bigl(\ddot{\varepsilon} +\dot{\varepsilon}^2\bigr)
    e^{2\varepsilon}X^2 -\frac{\rho_0}{2}
    \bigl(\dot{\varepsilon}^2 -\ddot{\varepsilon}\bigr)
    e^{-2\varepsilon}Y^2 +C(t),
\end{equation}
where $C(t)$ is an arbitrary function of time. Here and in what follows, the time dependence of $\varepsilon$ is omitted in order to simplify the notation. The free-boundary pressure condition \eqref{boundary cond} on $\partial\Omega_0 = \left\{ (X,Y)\in\mathbb{R}^2 : X^2+Y^2=1\right\}$ is satisfied if and only if the coefficients of $X^2$ and
$Y^2$ in \eqref{eq:material-pressure-general} are equal and coincide with $-C(t)$, which is precisely condition \eqref{eq:coefficient-equality}. We now want to investigate the maximal interval of existence of the geodesic given by the extensional motion \eqref{eq:extensional-placement}, that is, the maximal interval of existence of the solution $\varepsilon$ of the second order ODE \eqref{eq:coefficient-equality} with initial conditions $\varepsilon(0)=0$ and $\dot{\varepsilon}(0)=1$. Note that \eqref{eq:coefficient-equality} can be written as
\begin{equation*}
\label{eq:epsilon-geodesic-ode}
    \ddot{\varepsilon}
    =
    \dot{\varepsilon}^{\,2}
    \frac{e^{-4\varepsilon}-1}
         {e^{-4\varepsilon}+1}
\end{equation*}
or, equivalently,
\begin{equation*}
\label{eq:epsilon-geodesic-tanh}
    \ddot{\varepsilon}
    =
    -\dot{\varepsilon}^{\,2}\tanh(2\varepsilon).
\end{equation*}
Introducing the new variable $v(t):=\dot{\varepsilon}(t)$, the Cauchy problem reads
\begin{equation}
\label{eq:epsilon-v-system}
    \begin{cases}
        \dot{\varepsilon}=v,\\
        \dot v=-v^2\tanh(2\varepsilon),\\
        \varepsilon(0)=0,\\
        v(0)=1.
    \end{cases}
\end{equation}
The vector field
\[
    \bm{G}(\varepsilon,v)
    :=
    \begin{pmatrix}
        v\\
        -v^2\tanh(2\varepsilon)
    \end{pmatrix}
\]
is smooth on $\mathbb{R}^2$. Therefore, the classical Cauchy--Lipschitz theorem ensures the existence of a unique maximal solution
\[
(\varepsilon,v)\colon (T_-,T_+)\longrightarrow\mathbb{R}^2,
\]
where $-\infty\leq T_-<0<T_+\leq+\infty$. We now prove that the solution is global. First of all, note that, along a solution of \eqref{eq:epsilon-v-system}, one has
\begin{align*}
    \frac{d}{dt}
    \left(
        v^2\cosh(2\varepsilon)
    \right)
    &=
    2v\dot v\cosh(2\varepsilon)
    +
    2v^2\sinh(2\varepsilon)\dot{\varepsilon}
    \\
    &=
    2v
    \left[
        -v^2\tanh(2\varepsilon)
    \right]
    \cosh(2\varepsilon)
    +
    2v^2\sinh(2\varepsilon)v
    \\
    &=
    -2v^3\sinh(2\varepsilon)
    +
    2v^3\sinh(2\varepsilon)=0.
\end{align*}
Consequently, the quantity $v(t)^2\cosh(2\varepsilon(t))$ is constant in time and hence, by the initial conditions,
\begin{equation}
\label{eq:first-integral-epsilon}
    v(t)^2\cosh(2\varepsilon(t))=1
    \qquad
    \text{for every }t\in(T_-,T_+).
\end{equation}
Since $\cosh(2\varepsilon)\geq 1$,
equation \eqref{eq:first-integral-epsilon} implies
\begin{equation}
\label{eq:v-bound}
    |v(t)|\leq 1
    \qquad
    \text{for every }t\in(T_-,T_+).
\end{equation}
Moreover, because $\dot{\varepsilon}=v$ and
$\varepsilon(0)=0$, the fundamental theorem of calculus gives
\[
    \varepsilon(t)
    =
    \int_0^t v(s)\,ds.
\]
Using \eqref{eq:v-bound}, we obtain
\begin{equation}
\label{eq:epsilon-linear-bound}
|\varepsilon(t)|\leq \int_0^t |v(s)|ds
\leq|t|.
\end{equation}
We can now apply the standard continuation criterion for ordinary
differential equations. Suppose, by contradiction, that
$T_+<+\infty$. Choose any $\bar t\in(0,T_+)$. For every $t\in(\bar t,T_+)$, estimates
\eqref{eq:v-bound} and \eqref{eq:epsilon-linear-bound} give
\[
    |\varepsilon(t)|\leq T_+,
    \qquad
    |v(t)|\leq 1.
\]
Therefore,
\[
    \bigl(\varepsilon(t),v(t)\bigr)
    \in
    [-T_+,T_+]\times[-1,1]
    \qquad
    \text{for every }t\in(\bar t,T_+).
\]
Hence the solution remains trapped in a compact
subset of the domain as $t$ approaches $T_+$. By the continuation
theorem (see for example \cite{Hartman64}), it can therefore be prolonged beyond $T_+$, contradicting the maximality of $(T_-,T_+)$. It follows that $T_+=+\infty$. The argument at the left endpoint is analogous.
Therefore, $T_-=-\infty$. We conclude that the solution of
\eqref{eq:epsilon-v-system} is uniquely defined for every
$t\in\mathbb{R}$. Degenerations may occur only asymptotically, at infinite
time, but for every finite time both semi-axes remain strictly
positive.

The numerical solution of the differential problem \eqref{eq:epsilon-v-system} shows that the disk is progressively compressed in one direction, while it is indefinitely stretched in the other, with a monotonically decreasing deformation rate, as shown in Figure~\ref{fig:unica}.
\end{example}

\begin{figure}
\includegraphics[width=\textwidth]{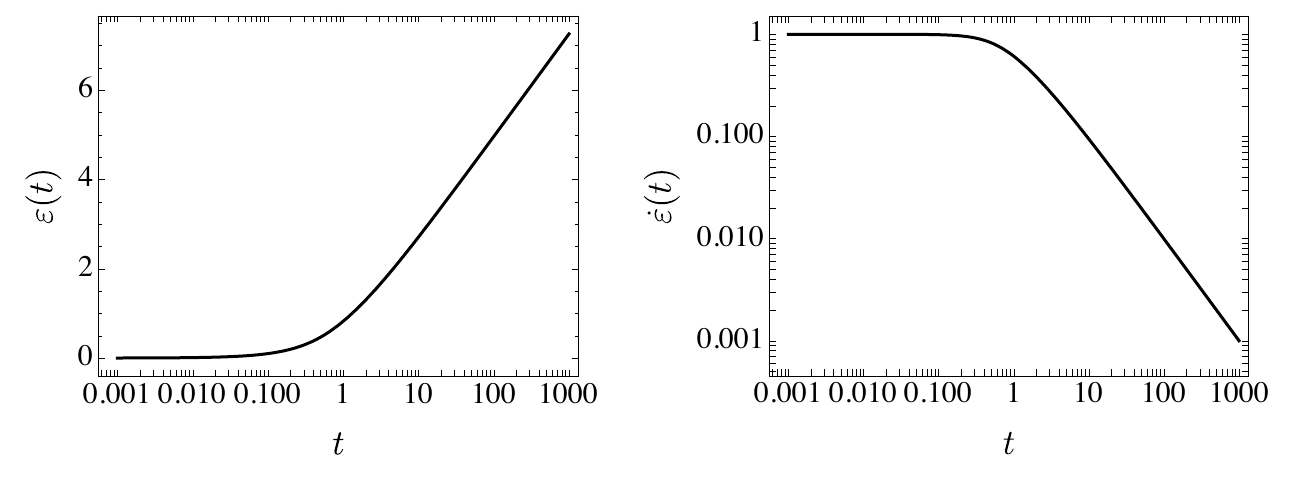}
\caption{Numerical solution of the differential problem \eqref{eq:epsilon-v-system}. The plot axes are chosen to highlight the asymptotic logarithmic growth of $\varepsilon(t)$ and the corresponding $1/t$ decay of $\dot{\varepsilon}(t)$.}\label{fig:unica}
\end{figure}

\begin{example}
Let now the material manifold, the initial deformation and the initial velocity be as in the previous example, and consider the manifold $\DSP$ of orientation-preserving deformations, endowed with the kinetic energy metric \eqref{metric}.
In the compressible case, where no volume-preserving constraint is imposed, the geodesic equation reduces to
\[
\ddot{\pl}(\bm{X},t)=0
\]
for every $\bm{X}\in\Omega_0$. The unique solution of the geodesic equation is therefore
\begin{equation}
\label{eq:compressible-affine-solution}
\begin{split}
    \bm{\varphi}(\bm{X},t)
    &=
    \bm{\varphi}_0(\bm{X})
    +
    t\bm{V}_0(\bm{X})=
    \begin{pmatrix}
        (1+t)X\\
        (1-t)Y
    \end{pmatrix}.
\end{split}
\end{equation}
The deformation gradient is
\begin{equation*}
\label{eq:compressible-F}
    \F(\bm{X},t)
    =
    \begin{pmatrix}
        1+t & 0\\
        0 & 1-t
    \end{pmatrix}.
\end{equation*}
Therefore, the Jacobian determinant is
\begin{equation*}
\label{eq:compressible-J}
    J(\bm{X},t) =
    (1+t)(1-t)
    =
    1-t^2.
\end{equation*}
Hence, for every $-1<t<1$, one has $J(\bm{X},t)>0$ and therefore $\pl(\cdot,t)\in\DSP$. The current configuration is the ellipse
\begin{equation*}
\label{eq:compressible-current-ellipse}
    \Omega_t
    =
    \bm{\varphi}_t(\Omega_0)
    =
    \left\{
        (x,y)\in\mathbb{R}^2
        :
        \frac{x^2}{(1+t)^2}
        +
        \frac{y^2}{(1-t)^2}
        <1
    \right\}.
\end{equation*}
Its horizontal and vertical semi-axes are respectively $R_x(t)=1+t$ and $R_y(t)=1-t$. Thus, for $0<t<1$, the disk is stretched in the horizontal direction and compressed in the vertical direction. Contrary to the
incompressible example, however, these two deformations do not
compensate each other: the area of the ellipse is
\[
    |\Omega_t|
    =
    \pi(1+t)(1-t)
    =
    \pi(1-t^2),
\]
which tends to zero as $t$ tends to $1$. At the critical time $t=1$, one has
\[
\varphi(\bm{X},1)=
    \begin{pmatrix}
        2X\\
        0
    \end{pmatrix}
\]
for every $\bm{X}\in\Omega_0$, that is, the disk collapses in the horizontal segment
\begin{equation*}
\label{eq:compressible-limit-segment}
    \bm{\varphi}_1(\Omega_0)
    =
    \left\{
        (x,0)\in\mathbb{R}^2
        :
        -2<x<2
    \right\}.
\end{equation*}
All the material points lying on the same vertical segment of
the reference disk are mapped onto a single spatial point. The deformation gradient has rank one and the placement map is no longer injective. Unlike the isotropic contraction described in Example \ref{ex:isotropic contraction}, the present degeneration
is anisotropic: only one material direction collapses, while the other
remains nondegenerate. Moreover, comparison with Example \ref{ex:incompressible} shows the effect of the
volume-preserving constraint: the incompressible motion with the same
initial data exists for all times, whereas its unconstrained counterpart
loses admissibility in finite time.
\end{example}

\subsection*{Acknowledgments}

The authors would like to thank Marco Degiovanni and Luis Garc\'ia-Naranjo for useful comments and discussions about the content of this work.

\subsection*{CRediT author statement}

\textbf{Francesca Berlinghieri:} Conceptualization, Formal analysis, Investigation, Writing - Original Draft.
\textbf{Giulio G.\ Giusteri:}
Conceptualization, Supervision, Writing - Review \& Editing.

\subsection*{Declaration of competing interests}
The authors declare no competing interests.

\bibliographystyle{plain}
\bibliography{refs}
\end{document}